%% file: main.tex
\pdfoutput=1

\documentclass[english,11pt]{article}
\include{macros}

\begin{document}

\title{\Large Distribution-Free Predictive Inference under Unknown Temporal Drift}\blfootnote{Author names are sorted alphabetically.}

\author{Elise Han\thanks{Department of Computer Science, Columbia University. Email: \texttt{lh3117@columbia.edu}.}
	\and Chengpiao Huang\thanks{Department of IEOR, Columbia University. Email: \texttt{chengpiao.huang@columbia.edu}.}
	\and Kaizheng Wang\thanks{Department of IEOR and Data Science Institute, Columbia University. Email: \texttt{kaizheng.wang@columbia.edu}.}
}

\date{This version: June 2024}

\maketitle

\begin{abstract}
Distribution-free prediction sets play a pivotal role in uncertainty quantification for complex statistical models. Their validity hinges on reliable calibration data, which may not be readily available as real-world environments often undergo unknown changes over time. In this paper, we propose a strategy for choosing an adaptive window and use the data therein to construct prediction sets. The window is selected by optimizing an estimated bias-variance trade-off. We provide sharp coverage guarantees for our method, showing its adaptivity to the underlying temporal drift. We also illustrate its efficacy through numerical experiments on synthetic and real data.
\end{abstract}

\input{main_intro}

\input{main_setup}

\input{main_methods}

\input{main_theory}

\input{main_experiments}

\input{main_discussions}

\section*{Acknowledgement}
Elise Han, Chengpiao Huang, and Kaizheng Wang's research is supported by an NSF grant DMS-2210907 and a startup grant at Columbia University.

\newpage 
\appendix

\input{appendix_proof}

\input{appendix_technical}

\input{appendix_experiments}

{
\bibliographystyle{ims}
\bibliography{bib}
}

\end{document}

%% file: main_intro.tex
\section{Introduction}\label{sec-intro}

Practitioners increasingly deploy sophisticated prediction models such as deep neural networks into real-world applications. Due to their complex structures, these models are generally accessed as black boxes. To assess their reliability and safeguard against potential errors, it is important to quantify the uncertainty in their outputs. \emph{Predictive inference} is a popular methodology for this purpose. It takes as input a prediction algorithm and \emph{calibration data}, and outputs a \emph{prediction set} that contains the true outcome with a prescribed probability. The validity of the prediction set hinges on the assumption that the calibration data truthfully represents the underlying environment. However, this assumption is frequently violated in practice, where the data distribution may drift over time. Integrating data from both current and historical periods to construct faithful prediction sets remains a significant challenge. Despite a large body of literature on learning under distribution drift over the past two decades \citep{HSe09,MMM12,BGZ15,HKY15,MUp23,HWa23}, statistical inference within this context is much less explored.

\paragraph{Main contributions.} In this paper, we study predictive inference under temporal distribution drift via quantile estimation. We develop an adaptive rolling window strategy to estimate the population quantile of a drifting distribution. This strategy in turn leads to an algorithm for building distribution-free prediction sets that wraps around any black-box point prediction algorithm. We further provide theoretical guarantees and numerical experiments to show that our algorithm adapts to unknown temporal drifts.

\paragraph{Related works.} One of the most powerful tools for distribution-free predictive inference is \emph{conformal prediction}, also called \emph{conformal inference} \citep{PPV02,VGS05,SVo08,LGR18,ABa23}. Under the assumption of data exchangeability, it takes a black-box point prediction model and uses the empirical quantile of \emph{conformity scores} to construct distribution-free prediction sets. Our method shares the same distribution-free and model-free features. When the data distribution undergoes temporal drift, the exchangeability assumption no longer holds. A number of works have considered extending the methodology to handle distribution shifts. Most of them assume that the calibration data are independent samples from a fixed source distribution \citep{TFC19,LCa21,PRa21,FBA22,PLS22,PDL22,QDT23,YKT24,SPL24,PXi24} or from a fixed number of source distributions \citep{BBa24,YGK24}, and consider only covariate shift or label shift. Their methods do not apply when the distribution keeps drifting over time.

Several recent works have studied conformal prediction under temporal drift. \cite{BCR23} proposes a general data weighting scheme but does not provide a principled approach for choosing the weights. \cite{LTS22,XXi23,LMa24} consider time series data with temporal dependence, but typically impose certain stationarity assumptions on the data. \cite{GCa21,GCa24,ZFG22,ACT23,ABB24,YCL24} perform online updates on some hyperparameters such as the miscoverage level and the quantile tracker, and provide long-term coverage guarantees. In contrast, our work focuses on achieving good coverage at a specific point in time, by utilizing \emph{offline} calibration data collected from the past. It can be used as a sub-routine for online predictive inference, and its coverage guarantees at individual time points translate to long-term coverage guarantees.

A widely used approach for adapting to temporal drift is \emph{rolling window} \citep{BGa07,HKY15,MMM12,MUp23,HWa23}. Our method adaptively selects a look-back window by optimizing an estimated bias-variance trade-off. It is inspired by the \emph{Goldenshluger-Lepski method} \citep{GLe08} for bandwidth selection in non-parametric estimation. The latter has been adapted for mean estimation under temporal drift in \cite{HHW24}. Compared with that work, we consider the more challenging problem of quantile estimation where the estimation error is measured by a nonconvex loss. Additionally, our approach is distribution-free, and does not make any boundedness assumption.

The methodology of our work shares similarities with but is different from conformal prediction. In particular, our work is based on estimating the population quantile, while conformal prediction utilizes properties of the empirical quantile under exchangeability. Our approach allows us to derive a stronger coverage guarantee called \emph{probably approximately correct (PAC)} or \emph{training-conditional coverage} \citep{Vov12,KJL20,BAL21,PSI21,BBa23,YAK24,LMa24,PXi24}. It requires that conditioned on the training and calibration data, the constructed prediction set has good coverage with high probability. We note that \cite{PDL22,QDT23,YKT24,SPL24} develop methods for such guarantees under covariate shift or label shift, but they assume that the calibration data comes from the same source distribution, and thus do not apply to drifting data.

\paragraph{Notation.} Let $\ZZ_+=\{1,2,...\}$ be the set of positive integers. For $n\in\ZZ_+$, let $[n]=\{1,2,...,n\}$. For $x\in\RR$, define $x_+ = \max\{x,0\}$. For non-negative sequences $\{a_n\}_{n=1}^{\infty}$ and $\{b_n\}_{n=1}^{\infty}$, we write $a_n=\cO(b_n)$ if there exists $C>0$ such that for all $n\in\ZZ_+$, $a_n \leq C b_n$. Unless otherwise stated, $a_n\lesssim b_n$ also represents $a_n=\cO(b_n)$. We write $a_n\asymp b_n$ if both $a_n\lesssim b_n$ and $b_n\lesssim a_n$ hold. For $x\in\RR$, we use $\delta_x$ to denote the Dirac measure at $x$. The total variation distance between two probability distributions $\distP$ and $\distQ$ is denoted by $\TV(\distP,\distQ)$. We use $\Unif(a,b)$ to denote the uniform distribution over the interval $[a,b]$. For an event $A$, we write $\ind(A)$ as its binary indicator.

\paragraph{Outline.} The rest of the paper is organized as follows. \Cref{sec-setup} formally describes the problem setup. \Cref{sec-methods} introduces our rolling window strategy for quantile estimation and predictive inference. \Cref{sec-theory} presents the theoretical guarantees. \Cref{sec-experiments} demonstrates the strong performance of our algorithms on synthetic and real data. \Cref{sec-discussions} concludes the paper and discusses future directions.

%% file: main_setup.tex
\section{Problem setup}\label{sec-setup}

In this section, we formally state the problem of predictive inference under temporal drift and then reduce it to quantile estimation.

\subsection{Predictive inference}

Let $\cX$ and $\cY$ be covariate and response spaces, respectively. At each time $t\in\ZZ_+$, we are given a predictive model $\widehat\mu_t:~ \cX \to \cY$ designed for the current probability distribution $\distP_t$ over $\cX\times\cY$. As the environment evolves over time, $\{ \distP_j \}_{j=1}^{t}$ can be different. For a new test point $(\bx_t,y_t)$ drawn from $\distP_t$, the model $\widehat\mu_t$ outputs a \emph{point prediction} $\widehat\mu_t(\bx_t)$ of the true label $y_t$. However, the model could make errors and the new data is random. To quantify the uncertainty of the prediction, we collect a batch of $B_t \geq 1$ i.i.d.~\emph{calibration data} $\datasetB_t = \{ (\bx_{t,i},y_{t,i})  \}_{i=1}^{B_t}$ from $\distP_t$. Based on the model $\widehat\mu_t$ and all the historical calibration data $\{ \datasetB_j \}_{j=1}^t$, we wish to construct a \emph{prediction set} that covers $y_t$ with a prescribed probability. This motivates the following problem:

\begin{problem}[Predictive inference]\label{prob-predictive-inference}
Given a constant $\alpha\in(0,1)$, a model $\widehat\mu_t$ and calibration data $\{\datasetB_j\}_{j=1}^t$, construct a set-valued mapping $\Chat_t:~ \cX \to 2^\cY$ such that for a new data point $(\bx_t,y_t)\sim \distP_t$,
\begin{align}
\PP\Big( y_t \in \Chat_t(\bx_t) \Big) \approx 1-\alpha.
\label{eqn-coverage-marginal}
\end{align}
\end{problem}

The probability in \eqref{eqn-coverage-marginal} is with respect to the randomness of $(\bx_t,y_t)$, $\{ \datasetB_j \}_{j=1}^t$, and $\widehat\mu_t$, which is common in the conformal prediction literature \citep{ABa23}. 
In practice, however, the latter two are often given to us and cannot be regenerated. Thus, ideally we would like the coverage probability of $\Chat_t ( \cdot )$ to be approximately $1 - \alpha$ for typical realizations of $\{\datasetB_j\}_{j=1}^t$ and $\widehat\mu_t$. That is, we want
\begin{align}
\PP\Big( y_t \in \Chat_t(\bx_t) \Bigm| \{\datasetB_j\}_{j=1}^t,~ \hat\mu_t \Big) \approx 1-\alpha
\label{eqn-pac-coverage}
\end{align}
to hold with high probability. Such guarantee, known as \emph{training-conditional coverage} or \emph{probably approximately correct (PAC) coverage} \citep{Vov12}, is stronger than \eqref{eqn-coverage-marginal}. We aim to establish \eqref{eqn-pac-coverage} under the following standard assumption.

\begin{assumption}[Independence]\label{assumption-independence}
The calibration datasets $\{ \datasetB_j \}_{j=1}^{t}$ are independent, and the model $\widehat\mu_t$ is trained independently of them.
\end{assumption}

For simplicity, we will treat $\widehat\mu_t$ as deterministic. This can be interpreted as conditioning on the training data and the training algorithm.

\subsection{Reduction to quantile estimation under temporal drift}\label{sec-setup-reduction}

We now present a reduction of \Cref{prob-predictive-inference} to quantile estimation using ideas from split conformal prediction \citep{ABa23}. Suppose we have a \emph{conformity score} function $s:~\cX \times \cY \to \RR$ so that $s(\bx_t, y_t)$ gauges the deviation of $y_t$ from our point prediction $\widehat{\mu}_t (\bx_t)$. Denote by $\cdf_t$ and $\quantile_t$ the cumulative distribution function (CDF) and a $(1 - \alpha)$-quantile of $s(\bx_t,y_t)$, respectively. Then, $s (\bx_t, y_t) \leq \quantile_t $ holds with probability at least $1 - \alpha$. If $\quantilehat_t$ is an estimate of $\quantile_t$ computed from the sample scores $\{ s (  \bx_{j,i} ,  y_{j, i} ) \}_{j\in [t], i \in [B_j]}$, then
\begin{align}
\Chat_t (\bx_t) = \{ y \in \cY :~ s (\bx_t, y) \leq \quantilehat_t \}
\label{eqn-prediction-set}
\end{align}
is a prediction set with coverage probability
\begin{align}
	\PP\big( y_t \in \Chat_t(\bx_t) \bigm| \{\datasetB_j\}_{j=1}^t \big)
	=
	\PP\big( s (\bx_t, y_t ) \leq \quantilehat_t \bigm| \{\datasetB_j\}_{j=1}^t \big)	
	=
	F_t(\quantilehat_t),
	\label{eqn-general-coverage}
\end{align}
conditioned on the calibration data. In light of the requirement \eqref{eqn-pac-coverage}, we want $|F_t(\quantilehat_t) - ( 1 - \alpha ) |$ to be small with high probability over the randomness of $\quantilehat_t$. In contrast, the goal \eqref{eqn-coverage-marginal} in conformal prediction translates to $\EE [ F_t ( \quantilehat_t ) ] \approx 1 - \alpha$.

For regression problems with $\cY = \RR$, a natural choice of $s(\bx, y)$ is the absolute residual $| y - \widehat{\mu}_t (\bx) |$, whose associated prediction set \eqref{eqn-prediction-set} has the form
\begin{align}
[ \widehat{\mu}_t (\bx_t) - \quantilehat_t,~
\widehat{\mu}_t (\bx_t) + \quantilehat_t ].
\label{eqn-prediction-interval}
\end{align}
Suppose that in addition to $\widehat\mu_t$, we also have an estimate $\widehat\sigma_t :~ \cX \to [0, +\infty) $ of $y_t$'s uncertainty given $\bx_t$ (e.g.,~conditional standard deviation or inter-quartile range). Then, we may use a studentized score $s (\bx, y) = | y - \widehat{\mu}_t (\bx) | / \widehat{\sigma}_t(\bx)$ and get a prediction interval
\begin{align*}
[\widehat{\mu}_t ( \bx_t ) - \quantilehat_t \widehat\sigma_t ( \bx_t ) ,~
\widehat{\mu}_t ( \bx_t ) + \quantilehat_t \widehat\sigma_t ( \bx_t )
].
\end{align*}
Compared to the constant-width interval \eqref{eqn-prediction-interval}, the new one can better capture heteroscedasticity.

The above discussion shows that our general goal can be reduced to the following problem.

\begin{problem}[Quantile estimation]\label{prob-quantile-estimation}
Let $\{\distQ_j\}_{j=1}^t$ be probability distributions over $\RR$, and let $\{\dataset_j\}_{j=1}^t$ be independent datasets, where $\dataset_j = \{u_{j,i}\}_{i=1}^{B_j}$ consists of $B_j$ i.i.d.~samples from $\distQ_j$. Given data $\{\dataset_j\}_{j=1}^t$ and a constant $\alpha\in(0,1)$, how to estimate the $(1-\alpha)$-quantile of $\distQ_t$? %In other words, we look for an estimate $\quantilehat_t$ such that $| F_t(\quantilehat_t) - (1-\alpha) |$ is small.
\end{problem}

Given a conformity score function $s$, we only need to solve \Cref{prob-quantile-estimation} with $\distQ_j = \mathrm{Law} ( s( \bx_t, y_t ) )$ and $u_{j,i} = s ( \bx_{j, i}, y_{j,i} )$. This amounts to quantile estimation under unknown temporal distribution shift. Thus, our methodology is fundamentally different from conformal prediction: the former targets the population quantile, while the latter utilizes the empirical quantile and crucially relies on the exchangeability of the samples.

We make the following minimal regularity assumption, which is satisfied when $\cY = \RR^d$, $y_t$ has a continuous conditional distribution given $\bx_t$, and $y\mapsto s(\bx,y)$ is continuous for every $\bx\in\cX$. No other condition, e.g., moment bound, is needed.

\begin{assumption}[Continuity]\label{assumption-AC}
The distributions $\{ \distQ_j \}_{j=1}^t$ have continuous CDFs.
\end{assumption}

%% file: main_methods.tex
\section{Methodology}\label{sec-methods}

In this section, we first present a bias-variance decomposition for the sample quantile under temporal drift. Based on that, we develop a solution to \Cref{prob-quantile-estimation} and its extension to \Cref{prob-predictive-inference}.

\subsection{A bias-variance decomposition for the sample quantile}

To set the stage, we introduce some notations.

\begin{definition}[Left and right quantiles]
Let $F:\RR\to[0,1]$ be a CDF. For $\gamma\in(0,1)$, the \emph{left $\gamma$-quantile} and \emph{right $\gamma$-quantile} of $F$ are defined by, respectively,
\[
\Qleft_{\gamma}(F) = \inf\{x\in\RR: F(x)\ge \gamma\}
\quad\text{and}\quad
\Qright_{\gamma}(F) = \inf\{x\in\RR: F(x) > \gamma\}.
\]
If $\cP$ is the probability distribution associated with $F$, we also write $\Qleft_{\gamma}(\cP)$ and $\Qright_{\gamma}(\cP)$ in place of $\Qleft_{\gamma}(F)$ and $\Qright_{\gamma}(F)$, respectively.
\end{definition}

The left quantile is the usual notion of quantile, while the right quantile plays an important role in our analysis. The relation $\Qleft_{\gamma} \leq \Qright_{\gamma}$ always holds.
\Cref{prob-quantile-estimation} amounts to estimating the left $(1-\alpha)$-quantile of $\distQ_t$, namely, $\Qleft_{1-\alpha} (F_t)$. A natural candidate is the left $(1-\alpha)$-quantile of the data $\{ \dataset_j \}_{j=t - k + 1}^t$ in some appropriately chosen look-back window $k$. This is a plug-in approach: the data defines an empirical CDF
\[
\widehat{F}_{t,k}(x) = \frac{1}{B_{t,k}} \sum_{j=t-k+1}^t \sum_{i=1}^{B_j} \ind\{ x \ge u_{j,i} \}\quad\text{where}\quad B_{t,k} = \sum_{j=t-k+1}^t B_j,
\]
and the estimator is 
\begin{equation}\label{eqn-empirical-quantile}
	\widehat{q}_{t,k} = \Qleft_{1-\alpha}(\widehat{F}_{t,k}).
\end{equation}

For any fixed $k$, the following result characterizes the approximation quality of $\widehat{q}_{t,k}$ through a bias-variance decomposition. Its proof is deferred to \Cref{sec-proof-thm-bias-variance-decomp}.

\begin{theorem}[Bias-variance decomposition]\label{thm-bias-variance-decomp}
Let Assumption \ref{assumption-AC} hold. Fix $k\in[t]$. Choose $\delta\in(0,1)$. Define
\begin{align}
& \phi (t, k) = \max_{t - k + 1 \leq j \leq t} \| F_j - F_t \|_{\infty}, \notag \\[4pt]
& \psi ( t , k, \delta ) =  \frac{5}{4} \sqrt{ \frac{ 2 \alpha ( 1 - \alpha ) \log (2 / \delta) }{ B_{t, k} } } + \frac{4 \log (2 / \delta ) }{ B_{t, k} }.
\label{eqn-psi}
\end{align}
With probability at least $1-\delta$,
\begin{equation}\label{eqn-bias-variance-decomp}
| 
	F_t(
	\widehat{q}_{t, k} 
	)
	- (1 - \alpha) | \leq  \phi (t, k) + \psi ( t , k, \delta ).
\end{equation}
\end{theorem}

Here $\phi(t,k)$ is the maximum Kolmogorov-Smirnov distance between $\distQ_t$ and $\distQ_{t-1},...,\distQ_{t-k+1}$, measuring the bias incurred by the distribution shift over the last $k$ periods. The term $\psi(t,k,\delta)$ is a Bernstein-type concentration bound \citep{BLM13}, representing the stochastic error in $\widehat{q}_{t,k}$. As $k$ increases, the bias grows and the stochastic error shrinks. %This reveals a bias-variance tradeoff phenomenon. 
Ideally, we would like to select $k$ that minimizes the bias-variance decomposition in \eqref{eqn-bias-variance-decomp}. However, this cannot be directly achieved since the bias term $\phi(t, k)$ involves the unknown distribution shift. Instead, we will construct an estimate of $\phi(t, k)$ for window selection.

%emphasize the advantage of Bernstein bound over uniform convergence when $\alpha$ is small?

\subsection{Adaptive rolling window for quantile estimation and predictive inference}

%By definition, $\cdfhat_{t, i} (\quantilehat_{t, i}) \approx 1 - \alpha$ for all $i \in [t]$. 
For two windows $i \le k$, if the environment is relatively stationary between time $t-i+1$ and time $t$ but there is significant cumulative shift between time $t-k+1$ and time $t-i$, then $\cdfhat_{t, i} (\quantilehat_{t, k})$ would be far from $\cdfhat_{t,i}(\quantilehat_{t,i}) \approx 1 - \alpha$. Thus, we can estimate the bias of window $k$ using the empirical CDF's $\{ \cdfhat_{t, i} (\quantilehat_{t, k}) \}_{i=1}^k$ over shorter windows. To start with, we show that $| \cdfhat_{t,i}(\widehat{q}_{t,k}) - (1-\alpha) |$ also admits a bias-variance decomposition similar to that of $| \cdf_t(\widehat{q}_{t,k}) - (1-\alpha) |$ in \Cref{thm-bias-variance-decomp}. The proof is in \Cref{sec-proof-thm-bias-variance-decomp-proxy}.

\begin{theorem}[Empirical bias-variance decomposition]\label{thm-bias-variance-decomp-proxy}
Let Assumption \ref{assumption-AC} hold. Suppose $1\leq i \leq k \leq t$, and choose $\delta \in (0, 1)$. With probability at least $1-\delta$,
\begin{equation}\label{eqn-bias-variance-decomp-proxy}
| 
\widehat{F}_{t, i} (  \widehat{q}_{t, k} ) 
- (1 - \alpha) | \leq 
\frac{12}{5} \phi (t, k) + \frac{6}{5} \psi (t, k, \delta/2) + \frac{4}{5} \psi (t, i, \delta/2).
\end{equation}
\end{theorem}

We see that \eqref{eqn-bias-variance-decomp-proxy} differs from \eqref{eqn-bias-variance-decomp} only by multiplicative constants and an additional variance term associated with the smaller window $i$. In light of \eqref{eqn-bias-variance-decomp-proxy} and inspired by the \emph{Goldenshluger-Lepski method} \citep{GLe08} for adaptive nonparametric estimation, we define a bias proxy
\begin{equation}\label{eqn-bias-proxy}
\widehat{\phi}(t,k,\delta)
= \frac{5}{12} \max_{i\in[k]} \bigg( \big| \widehat{F}_{t,i}(\widehat{q}_{t,k}) - (1-\alpha) \big| - \bigg[ \frac{6}{5} \psi (t,k,\delta/2) + \frac{4}{5} \psi (t,i,\delta/2) \bigg]    \bigg)_+.
\end{equation}

Essentially, $\widehat{\phi}(t,k,\delta)$ teases out the bias term $\phi(t,k)$ in \eqref{eqn-bias-variance-decomp-proxy} by subtracting the stochastic errors $ \frac{6}{5} \psi (t,k,\delta/2) + \frac{4}{5} \psi (t,i,\delta/2)$ associated with $\widehat{q}_{t,k}$ and $\widehat{F}_{t,i}$. We take maximum over all smaller windows $i\in[k]$ so that we do not miss any distribution shifts over the last $k$ periods. As a direct consequence of \eqref{eqn-bias-variance-decomp-proxy}, $0 \le \widehat{\phi}(t,k,\delta) \le \phi(t,k)$ holds with high probability. It is then natural to choose a window that minimizes the approximate bias-variance decomposition $ \widehat{\phi}(t,k,\delta) + \psi(t,k,\delta)$. This leads to \Cref{alg-quantile}.

\begin{algorithm}[h]
	\begin{algorithmic}
	\STATE {\bf Input:} Datasets $\{ \dataset_j \}_{j=1}^t$ with $\dataset_j = \{ u_{j, i} \}_{i=1}^{B_j}$, miscoverage level $\alpha$ and hyperparameter $\delta'$.
	%\STATE Let $\delta' = \delta/(4t)$.
	\FOR{$k = 1,\cdots, t$}
	\STATE Compute $\widehat{q}_{t,k}$, $\psi  (t , k , \delta')$, $\psi  (t , k , \delta'/2)$ and $\widehat{\phi}(t,k,\delta')$ according to \eqref{eqn-empirical-quantile}, \eqref{eqn-psi} and \eqref{eqn-bias-proxy}.
	\ENDFOR
	\STATE Choose any $	\widehat{k} \in \argmin_{ k \in [t] }  \{ \widehat\phi (t, k, \delta')  + \psi (t , k , \delta')  \}$.
	\RETURN $\widehat{q}_{t,\widehat{k}}$. 
	\caption{Adaptive rolling window for quantile estimation}
	\label{alg-quantile}
	\end{algorithmic}
\end{algorithm}

As we have described in \Cref{sec-setup-reduction}, \Cref{alg-quantile} for quantile estimation immediately yields a procedure for constructing predictive sets. See \Cref{alg-predictive-inference} for a detailed description.

\begin{algorithm}[h]
	\begin{algorithmic}
	\STATE {\bf Input:} Calibration datasets $\{ \datasetB_j \}_{j=1}^t$ with $\datasetB_j = \{ ( \bx_{j, i}, y_{j,i} ) \}_{i=1}^{B_j}$, conformity score $s:~ \cX \times \cY \to \RR$, miscoverage level $\alpha$ and hyperparameter $\delta'$.
	\STATE Let $u_{j,i} = s( \bx_{j, i}, y_{j,i} )$ and set $\dataset_j = \{u_{j,i} \}_{i=1}^{B_j} $.
	\STATE Run \Cref{alg-quantile} with input $\{\dataset_j\}_{j=1}^t$, $\alpha$ and $\delta'$ to obtain $\widehat{q}_{t,\widehat{k}}$.
	\RETURN $\Chat_t(\cdot) =   \{ y \in \cY :~ s ( \cdot , y) \leq \widehat{q}_{t,\widehat{k}} \}$. 
	\caption{Adaptive rolling window for predictive inference}
	\label{alg-predictive-inference}
	\end{algorithmic}
\end{algorithm}

%% file: main_theory.tex
\section{Theoretical analysis}\label{sec-theory}

We now present theoretical guarantees for \Cref{alg-quantile,alg-predictive-inference}, showing their adaptivity to the unknown distribution drift over time.

\begin{theorem}[Oracle inequality for \Cref{alg-quantile}]\label{thm-GL}
	Choose $\delta\in(0,1)$ and take $\delta' = \delta / (4t^2)$ in \Cref{alg-quantile}. With probability at least $1-\delta$, \Cref{alg-quantile} outputs $\widehat{q}_{t,\widehat{k}}$ satisfying
	\begin{equation}
	\big| F_t(\widehat{q}_{t,\widehat{k}}) - (1-\alpha) \big| \lesssim \min_{k\in[t]} \left\{ \phi(t,k) + \psi(t,k) \right\}
		\lesssim \min_{k\in[t]} \left\{ \phi(t,k) + \sqrt{ \frac{\alpha (1 - \alpha)}{B_{t, k}} }  
		+ \frac{1}{B_{t, k}}
		\right\}
		.
	\end{equation}
	Here $\lesssim$ only hides a logarithmic factor of $t$ and $\delta^{-1}$. 
\end{theorem}

According to \Cref{thm-GL}, \Cref{alg-quantile} chooses a window $\widehat{k}$ that is near-optimal for the bias-variance trade-off in \eqref{eqn-bias-variance-decomp}, without any knowledge of the underlying distribution drift. We illustrate this oracle property using the following examples.

\begin{example}[Change point]\label{eg-changepoint}
	Suppose that the environment is stationary in the last $K$ periods for some $K\in[t-1]$ but has been very different before, i.e.~$\distQ_{t-K}\neq \distQ_{t-K+1} = \cdots = \distQ_t$. If $K$ were known, then a natural estimate of the $(1-\alpha)$-quantile of $\distQ_t$ would be $\widehat{q}_{t,K}$. \Cref{thm-GL} shows that the estimation quality of $\widehat{q}_{t,\widehat{k}}$ is comparable to $\widehat{q}_{t,K}$ which uses $B_{t,K}$ i.i.d.~samples: up to logarithmic factors,
	\[
	\big| F_t(\widehat{q}_{t,\widehat{k}}) - (1-\alpha) \big|  
	\lesssim   \phi(t,K) + \psi(t,K)  \lesssim \sqrt{\frac{\alpha(1-\alpha)}{B_{t,K}}} +  \frac{1}{B_{t, K}}.
	\]
	In this case, \Cref{alg-quantile} adapts to the local stationarity.
\end{example}

\begin{example}[Bounded drift]\label{eg-bounded-drift}
	Suppose that the distribution drift between consecutive periods is bounded, in the sense that $\sup_{j \geq 1} \|F_{j+1} - F_j\|_{\infty} \le   \Delta$ holds for some $0 < \Delta \leq 1$. For simplicity, we further assume $B_j=1$ for all $j$. In this case, $\phi(t,k) \le (k-1)\Delta$, so the bias-variance decomposition in \Cref{thm-bias-variance-decomp} becomes
	\[
	| F_t (\widehat{q}_{t,k}) - (1-\alpha) | \lesssim (k-1)\Delta + \sqrt{\frac{\alpha(1-\alpha)}{k}} + \frac{1}{k}
	\]
	up to a logarithmic factor. If $\Delta$ were known, then one could pick the optimal window size $k^* \asymp \Delta^{-2/3}$, which gives an error of $\cO(\Delta^{1/3})$. \Cref{thm-GL} shows that without knowing $\Delta$, \Cref{alg-quantile} picks $\widehat{k}$ whose quality is comparable to $k^*$.
\end{example}

Below we provide a sketch of proof for \Cref{thm-GL}. The full proof is given in \Cref{sec-proof-thm-GL}.

\begin{proof}[\bf Proof sketch for \Cref{thm-GL}]
To keep the notations simple and focus on the main ideas, we suppress the parameters $t$ and $\delta$ in the expressions of $\phi$, $\widehat{\phi}$ and $\psi$. We will prove that with high probability,
	\[
	\big| F_t(\widehat{q}_{t,\widehat{k}}) - (1-\alpha) \big| \lesssim  \phi(k) + \psi(k),\quad\forall k\in[t].
	\]
	By \Cref{thm-bias-variance-decomp-proxy} we know that $0\le \widehat{\phi}  \le \phi $ with high probability. For every $k\in\{\widehat{k}+1,\cdots,t\}$, with high probability,
	\begin{align*}
		|F_t(\widehat{q}_{t,\widehat{k}}) - (1-\alpha)|
		&\le \phi( \widehat{k}) + \psi( \widehat{k} ) \tag{by \Cref{thm-bias-variance-decomp}}  \\[4pt]
		&\le \phi( k) + \widehat{\phi}( \widehat{k} ) + \psi( \widehat{k} )  \tag{by the motonicity of $\phi$ and the fact $\widehat{\phi}\ge 0$} \\[4pt]
		&\le \phi( k) + \widehat{\phi}( k ) + \psi( k ) \tag{by the definition of $\widehat{k}$} \\[4pt]
		&\lesssim	\phi( k) + \psi( k ). \tag{by $\widehat{\phi}\le\phi$}
	\end{align*}
	
It remains to consider $k\in [ \widehat{k} ]$. Choose any $k$ in this range. With high probability,
	\begin{align}
		|\widehat{F}_{t,k}(\widehat{q}_{t,\widehat{k}}) - (1-\alpha)|
		&\lesssim 
		\widehat{\phi}( \widehat{k} ) + \psi ( \widehat{k} ) + \psi ( k ) \tag{by the definition of $\widehat{\phi}$} \\[4pt]
		&\le 
		\widehat{\phi}( k ) + \psi( k ) + \psi( k ) \tag{by the definition of $\widehat{k}$} \\[4pt]
		&\lesssim
		\phi( k) + \psi( k ). \tag{by $\widehat{\phi}\le\phi$}
	\end{align}
	Let $r_k \asymp \phi(k) + \psi(k)$ denote the right-hand side of the inequality above. The bound $|\widehat{F}_{t,k}(\widehat{q}_{t,\widehat{k}}) - (1-\alpha)| \le r_k$ implies that the final estimate $ \widehat{q}_{t,\widehat{k}} $ is sandwiched between the left $(1 - \alpha - r_k)$-quantile and the right $(1 - \alpha + r_k)$-quantile of $\widehat{F}_{t,k} $. Hence,
	\[
\cdf_t \big( \Qleft_{1-\alpha - r_k}  ( \widehat{F}_{t,k}  ) \big) 
\le \cdf_t ( \widehat{q}_{t,\widehat{k}} )
\le 
\cdf_t \big( \Qright_{1-\alpha + r_k}  ( \widehat{F}_{t,k}  ) \big).
\]
We use a concentration bound for empirical quantile under distribution shift (\Cref{lem-quantile} in \Cref{sec-technical}) to derive that with high probability,
\begin{align*}
& \cdf_t \big( \Qright_{1-\alpha + r_k} ( \cdfhat_{t, k} ) \big) - (1-\alpha + r_k) \lesssim
\phi (k) + \psi (k) , \\[4pt]
& (1-\alpha - r_k) - \cdf_t \big( \Qleft_{1-\alpha - r_k} ( \cdfhat_{t, k} ) \big) \lesssim
\phi (k) + \psi (k) .
\end{align*}
Combining the above estimates, we get
\[
| F_t(\widehat{q}_{t,\widehat{k}}) - (1-\alpha) |
\lesssim
\phi(k) + \psi (k) , \qquad \forall k \in [\widehat{k}] .
\]
This finishes the proof.
\end{proof}

\Cref{thm-GL} also translates to a theoretical guarantee for \Cref{alg-predictive-inference}. In particular, $\distQ_j$ and $\cdf_j$ are, respectively, the distribution and the CDF of the conformity score $s( \bx_j , y_j )$, where $(\bx_j,y_j)\sim\distP_j$. To facilitate interpretation, we may further upper bound the bias term by 
\[
\phi(t,k) = \max_{t-k+1\le j\le t} \| F_j - F_t \|_{\infty} \le \max_{t-k+1\le j\le t}\TV(\distP_j,\distP_t),
\] 
which is the maximum total variation distance between $\distP_t$ and $\distP_{t-1},...,\distP_{t-k+1}$. This establishes the following result. Its proof is omitted.

\begin{corollary}\label{cor-predictive-inference}
Let Assumption \ref{assumption-independence} hold.	Choose $\delta\in(0,1)$ and take $\delta' = \delta / (4t^2)$ in \Cref{alg-predictive-inference}. With probability at least $1-\delta$, \Cref{alg-predictive-inference} outputs $\Chat_t(\cdot)$ satisfying
	\[
	\left |\PP\big( y_t \in \Chat_t(\bx_t) \mid \{\datasetB_j\}_{j=1}^t \big) - (1-\alpha) \right| 
	\lesssim
	\min_{k\in[t]} \left\{ \max_{t-k+1\le j\le t}\TV(\distP_j,\distP_t) + \sqrt{\frac{\alpha(1-\alpha)}{B_{t,k}}}  
+	 \frac{1}{B_{t, k}}
	 \right\},
	\]
	where $\lesssim$ only hides a logarithmic factor of $t$ and $\delta^{-1}$. Consequently,
	\begin{align*}
		&\left |\PP\big( y_t \in \Chat_t(\bx_t) \big) - (1-\alpha) \right| 
		\\[4pt]&
		\le 
		\EE \left |\PP\big( y_t \in \Chat_t(\bx_t) \bigm| \{\datasetB_j\}_{j=1}^t \big) - (1-\alpha) \right| \\[4pt]
		&\lesssim
		(1-\delta)\min_{k\in[t]} \left\{ \max_{t-k+1\le j\le t}\TV(\distP_j,\distP_t) + \sqrt{\frac{\alpha(1-\alpha)}{B_{t,k}}} + \frac{1}{B_{t,k}} \right\} + \delta\max\{\alpha,1-\alpha\}.
	\end{align*}
\end{corollary}

%% file: main_experiments.tex
\section{Numerical experiments}\label{sec-experiments}

We test our proposed method on synthetic and real data, focusing on regression problems with $\cY = \RR$ and prediction intervals of the form \eqref{eqn-prediction-interval}. As a common phenomenon in learning theory, the non-asymptotic analysis provides sharp rates but not necessarily optimal constants. We make minor changes to the functions $\psi$ and $\widehat{\phi}$ in \Cref{alg-quantile} to correct for overly large constant factors. To improve computational efficiency, we only consider candidate window sizes that are powers of 2; see \Cref{alg-quantile-experiment} in \Cref{sec-algo-refined} for a detailed description. The theoretical guarantees in \Cref{sec-theory} continue to hold for this modification up to a constant factor. Our prediction intervals are constructed by \Cref{alg-predictive-inference}, with \Cref{alg-quantile-experiment} (instead of \Cref{alg-quantile}) as the sub-routine for quantile estimation. 

Throughout the experiments, we take $\alpha = 0.1$ to aim for prediction sets with $90\%$ coverage, and set the hyperparameter $\delta'$ to be $0.1$. For notational convenience, we denote our aforementioned method by $\ARW$ (adaptive rolling window). We will compare it with the following two families of benchmark approaches: 
\begin{itemize}
\item \Cref{alg-predictive-inference-fixed-window}, denoted by $\cV_k$. This is a non-adaptive, fixed-window version of \Cref{alg-predictive-inference}. We consider window sizes $k \in \{1,4,16,64,256,1024\}$. 

\item The nonexchangeable split conformal prediction method in Section 3.1 of \cite{BCR23}, with weights $w_i = \rho ^{t-i}$, $i\in[t]$ at time $t$. We refer to this method as $\cW_{\rho}$, and consider decay rates $\rho \in \{ 0.99,0.9,0.25 \}$.
\end{itemize}

\begin{algorithm}[h]
	\begin{algorithmic}
	\STATE {\bf Input:} Calibration datasets $\{ \datasetB_j \}_{j=1}^t$ with $\datasetB_j = \{ ( \bx_{j, i}, y_{j,i} ) \}_{i=1}^{B_j}$, conformity score $s:~ \cX \times \cY \to \RR$, miscoverage level $\alpha$ and window size $k$.
	\STATE Let $u_{j,i} = s( \bx_{j, i}, y_{j,i} )$. Compute $\widehat{q}_{t,k\wedge t}$ according to \eqref{eqn-empirical-quantile}.
	\RETURN $\Chat_t(\cdot) =   \{ y \in \cY :~ s ( \cdot , y) \leq \widehat{q}_{t,k\wedge t} \}$. 
	\caption{Predictive inference via fixed-window quantile estimation}
	\label{alg-predictive-inference-fixed-window}
	\end{algorithmic}
\end{algorithm}

All our results are averaged over 100 independent runs. The standard errors are less than $10\%$ of those mean values and thus omitted for space considerations. Our code is available at \url{https://github.com/eliselyhan/predictive-inference}.

\subsection{Synthetic data}

In the synthetic data experiment, we perform predictive inference for two problem instances, namely, Gaussian mean estimation and linear regression. In both instances, there are $T = 1000$ periods in total. In each period, we split the data into a training set for fitting predictors, and a calibration set for forming prediction sets. 

\paragraph{Gaussian mean estimation.} The mean estimation problem has $\cX=\varnothing$, $\cY = \RR$, and $\distP_t = \cN(\mu_t,1)$. We consider two patterns of the sequence $\{\mu_t\}_{t=1}^T$, given in \Cref{fig:true-means}. 
%, which we refer to as the stationary and the non-stationary cases. 
The left panel represents the stationary case where $\mu_1=\cdots=\mu_T$. The right panel corresponds to a highly non-stationary pattern generated using a randomized mechanism, which we describe in \Cref{sec-random-function}.

\begin{figure}[h]
	\centering
	\begin{subfigure}{
			\label{fig-true-means-stationary}
			\includegraphics[scale=0.55]{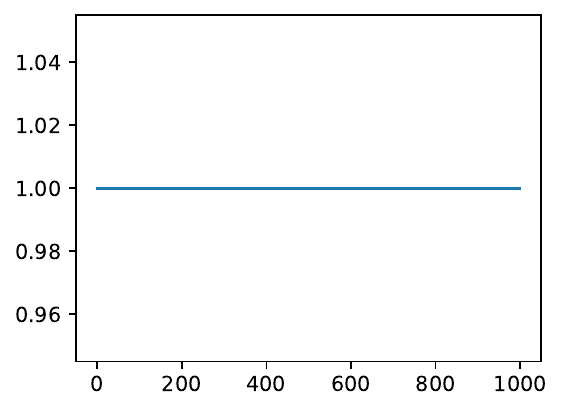}
		}
	\end{subfigure}
	\quad
	\begin{subfigure}{
			\label{fig-true-means-nonstationary}
			\includegraphics[scale=0.55]{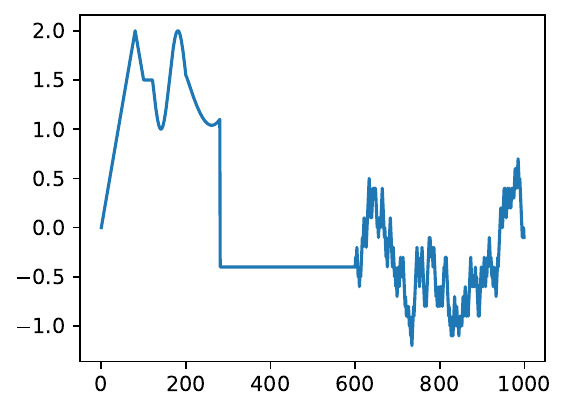}
		}
	\end{subfigure}
	\caption{True means $\mu_t$.}
	\label{fig:true-means}
\end{figure}

In each period $j\in[T]$, we generate i.i.d.~training samples $\datasetB_j^{\tr}$ and i.i.d.~calibration samples $\datasetB_j^{\ca}$ from $\distP_j$. The size $|\datasetB_j^{\tr}|$ of the training set $\datasetB_j^{\tr}$ is drawn uniformly over $\{ 1 , \cdots, 9 \}$, and $|\datasetB_j^{\tr}|=|\datasetB_j^{\ca}|$. At time $t\in[T]$, we consider point estimates $\widehat{\mu}_{t,k}\in\RR$ that compute moving averages of the data $\{\datasetB_j\}_{j=t-k\wedge t+1}^t$ from the last $k\in\{1,16,256,1024\}$ periods. For each model, we then compute the conformity scores using the calibration data.

Each of the candidate methods ($\ARW$ and the benchmark algorithms $\cV_k$ and $\cW_{\rho}$) produces a prediction interval $\widehat{C}_t$ of the form \eqref{eqn-prediction-interval}. We calculate the mean absolute error (MAE) between the true coverage and target level $1-\alpha$, given by 
\[
\text{MAE} = \frac{1}{T-100} \sum_{t=T-99}^{T} |\PP(z_t\in\widehat{C}_t\mid \widehat{C}_t) - (1-\alpha)|,
\]
where $z_t\sim\distP_t$, and the conditional probability $\PP(z_t\in \widehat{C}_t\mid \widehat{C}_t)$ can be computed using the CDF of the normal distribution. Following \cite{BCR23}, we average data only after the initial burn-in period with $100$ steps.

Tables \ref{tab:mean-stationary} and \ref{tab:mean-nonstationary} list the MAEs of different algorithms for the two patterns of $\{\mu_t\}_{t=1}^T$ respectively. The results are averaged over $100$ independent runs. In each setting, the best performance of benchmark algorithm is shown in boldface. Our algorithm $\ARW$ achieves comparable results while being agnostic to the underlying temporal drift. \Cref{fig:mean-est-bar} provides bar plots to visualize the MAEs of all algorithms for the prediction model $\widehat{\mu}_{t,1}$, which corresponds to the first rows of \Cref{tab:mean-stationary,tab:mean-nonstationary}.

\begin{table}[!h]
	\vskip 0.15in
	\centering
				\begin{tabular}{lccccccccccr}
					\toprule
					Training\\ Window & ARW & $\mathcal{W}_{0.99}$ & $\mathcal{W}_{0.9}$ & $\mathcal{W}_{0.5}$ & $\mathcal{W}_{0.99}$ & $\mathcal{V}_1$ & $\mathcal{V}_4$ & $\mathcal{V}_{16}$  & $\mathcal{V}_{64}$ & $\mathcal{V}_{256}$ & $\mathcal{V}_{1024}$ \\
					\midrule
					1 & 0.50 & {\bf 0.48} & {\bf 0.48} & 0.62 & 0.91 & 15.32 & 5.63 & 2.71 & 1.33 & 0.69 & {\bf 0.48} \\
					64 & 0.47 & {\bf 0.46} & {\bf 0.46} & 0.61 & 0.90 & 15.31 & 5.67 & 2.72 & 1.35 & 0.67 & {\bf 0.46} \\
					256 & 0.47 & {\bf 0.46} & {\bf 0.46} & 0.62 & 0.92 & 15.30 & 5.66 & 2.71 & 1.33 & 0.68 & {\bf 0.46} \\
					1024 & 0.47 & {\bf 0.45} & {\bf 0.45} & 0.61 & 0.91 & 15.30 & 5.66 & 2.71 & 1.33 & 0.68 & {\bf 0.45} \\
					\bottomrule
				\end{tabular}
	\vskip 0.1in
	\caption{Mean absolute error (in \%) of coverage for stationary mean estimation.}
	\label{tab:mean-stationary}
\end{table}

\begin{table}[h]
	\vskip 0.15in
	\centering
				\begin{tabular}{lccccccccccccr}
					\toprule
					Training\\ Window & ARW & $\mathcal{W}_{0.99}$ & $\mathcal{W}_{0.9}$ & $\mathcal{W}_{0.5}$ & $\mathcal{W}_{0.99}$ & $\mathcal{V}_1$ & $\mathcal{V}_4$ & $\mathcal{V}_{16}$  & $\mathcal{V}_{64}$ & $\mathcal{V}_{256}$ & $\mathcal{V}_{1024}$ \\
					\midrule
					1 & 3.28 & 7.24 & 7.26 & 7.50 & 7.71 & 15.32 & 5.65 & 2.99 & {\bf 2.81} & 4.41 & 7.24 \\
					64 & 2.53 & 7.34 & 7.37 & 7.65 & 7.88 & 15.28 & 5.70 & 3.06 & {\bf 2.23} & 3.68 & 7.34 \\
					256 & 3.04 & 7.73 & 7.76 & 7.97 & 8.17 & 15.34 & 5.72 & 3.18 & {\bf 3.10} & 4.54 & 7.73 \\
					1024 & 3.50 & 5.45 & 5.45 & 5.47 & 5.54 & 15.42 & 5.79 & {\bf 3.36} & 3.45 & 5.23 & 5.45 \\
					\bottomrule
				\end{tabular}
	\vskip 0.1in
	\caption{Mean absolute error (in \%) of coverage for non-stationary mean estimation.}
	\label{tab:mean-nonstationary}
\end{table}

\begin{figure}[h]
    \centering
    \begin{subfigure}{
            %\label{fig-combination-time-var1}
            \includegraphics[scale=0.45]{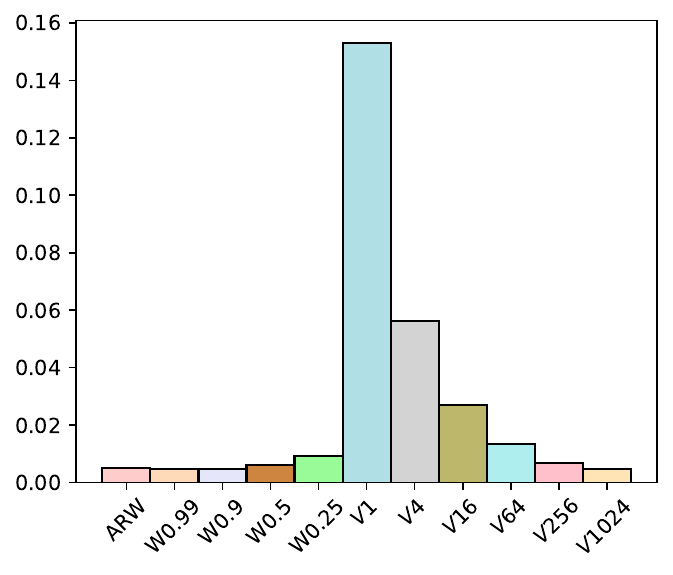}
        }
    \end{subfigure}
    \quad
    \begin{subfigure}{
            %\label{fig-combination-time-var10}
            \includegraphics[scale=0.45]{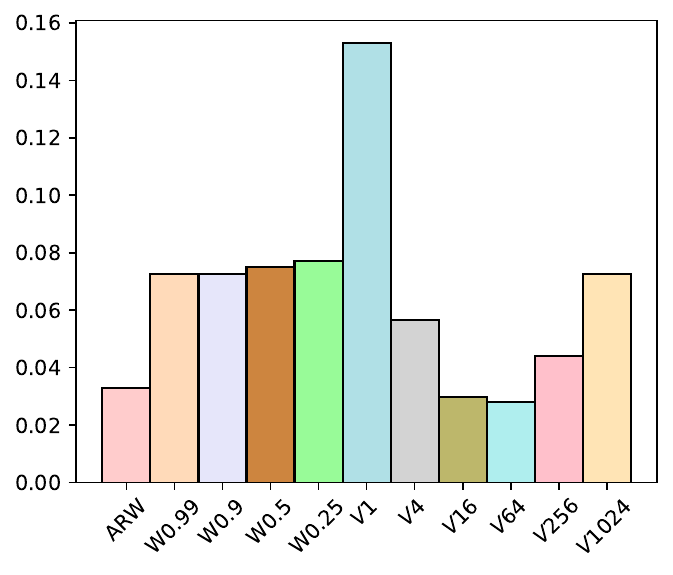}
        }
    \end{subfigure}
\caption{Mean absolute errors of coverage for mean estimation. Left: stationary case. Right: non-stationary case.}
\label{fig:mean-est-bar}
\end{figure}

\paragraph{Linear regression.} We set $\cX = \RR^5$. For each $t\in[T]$, the data distribution is $\distP_t = \distP_{t,\bx} \times \distP_{t,y\mid \bx}$, where $\distP_{t,\bx} = \cN(\bm{0},\bI_5)$ is the marginal distribution of the covariate $\bx$, and $\distP_{t,y\mid \bx} = \cN(\bx^\top\bbeta_t,1)$ is the conditional distribution of $y$ given $\bx$. As in the mean estimation problem, we consider two patterns of $\{\bbeta_t\}_{t=1}^T$. Let $\bm{1}_5 =(1,1,1,1,1)^\top\in\RR^5$. In the stationary pattern, $\bbeta_t = \frac{1}{5} \bm{1}_5$ for all $t\in[T]$. In the non-stationary pattern, $\bbeta_t = 2 \beta_t \bm{1}_5$, where $\{\beta_t\}_{t=1}^T$ follows the non-stationary pattern in the right panel of \Cref{fig:true-means}; see \Cref{sec-random-function} for a detailed description.

In each period $j\in[T]$, we generate i.i.d.~training samples $\datasetB_j^{\tr}$ and i.i.d.~calibration samples $\datasetB_j^{\ca}$ from $\distP_j$. The size $|\datasetB_j^{\ca}|$ of the calibration set $\datasetB_j^{\ca}$ is drawn uniformly over $\{ 1 , \cdots, 9 \}$, and $|\datasetB_j^{\tr}|=3|\datasetB_j^{\ca}|$. At time $t\in[T]$, we run ordinary least squares (OLS) to train a linear prediction model $\widehat{\mu}_{t,k}$, where $k\in\{1,16,256,1024\}$, and $\widehat{\mu}_{t,k}$ uses the training data from the last $k$ periods $\{\datasetB_j^{\tr}\}_{j=t-t\wedge k +1}^t$. We carry out the same procedures as in the mean estimation problem to compute conformity scores and form prediction intervals using $\ARW$, $\cV_k$ and $\cW_{\rho}$. The true coverage level $\cvrg_t = \PP(y_t\in\widehat{C}_t(\bx_t)\mid\widehat{C}_t)$, where $(\bx_t,y_t)\sim\distP_t$, is approximated by drawing another $1000$ independent samples from $\distP_t$.

In Tables \ref{tab:ols-stationary} and \ref{tab:ols-nonstationary}, we report the MAEs of different algorithms over $100$ independent runs. \Cref{fig:linreg-bar} provides a bar plot visualizing the MAEs of different algorithms using the OLS point predictor $\widehat{\mu}_{t,1}$.

\begin{table}[h]
	\vskip 0.15in
	\centering
				\begin{tabular}{lccccccccccr}
					\toprule
					Training\\ Window & ARW & $\mathcal{W}_{0.99}$ & $\mathcal{W}_{0.9}$ & $\mathcal{W}_{0.5}$ & $\mathcal{W}_{0.99}$ & $\mathcal{V}_1$ & $\mathcal{V}_4$ & $\mathcal{V}_{16}$  & $\mathcal{V}_{64}$ & $\mathcal{V}_{256}$ & $\mathcal{V}_{1024}$ \\
					\midrule
					1 & 0.90 & {\bf 0.90} & {\bf 0.90} & 0.96 & 1.15 & 15.44 & 5.69 & 2.80 & 1.53 & 1.00 & {\bf 0.90} \\
					64 & 0.91 & {\bf 0.90} & {\bf 0.90} & 0.96 & 1.16 & 15.43 & 5.68 & 2.80 & 1.53 & 1.00 & {\bf 0.90} \\
					256 & 0.90 & 0.90 & {\bf 0.89} & 0.96 & 1.15 & 15.44 & 5.68 & 2.79 & 1.52 & 1.00 & 0.90 \\
					1024 & 0.91 & {\bf 0.90} & {\bf 0.90} & 0.96 & 1.15 & 15.44 & 5.67 & 2.79 & 1.52 & 1.00 & {\bf 0.90} \\
					\bottomrule
				\end{tabular}
	\vskip 0.1in
	\caption{Mean absolute error (in \%) of coverage for stationary linear regression.}
	\label{tab:ols-stationary}
\end{table}

\begin{table}[h]
	\vskip 0.15in
	\centering
				\begin{tabular}{lccccccccccr}
					\toprule
					Training \\ Window & ARW & $\mathcal{W}_{0.99}$ & $\mathcal{W}_{0.9}$ & $\mathcal{W}_{0.5}$ & $\mathcal{W}_{0.99}$ & $\mathcal{V}_1$ & $\mathcal{V}_4$ & $\mathcal{V}_{16}$  & $\mathcal{V}_{64}$ & $\mathcal{V}_{256}$ & $\mathcal{V}_{1024}$ \\
					\midrule
					1 & 3.60 & 8.69 & 8.67 & 8.49 & 8.33 & 15.46 & 5.77 & {\bf 3.29} & 3.61 & 5.84 & 8.69 \\
					64 & 3.63 & 8.63 & 8.61 & 8.44 & 8.26 & 15.44 & 5.77 & {\bf 3.29} & 3.62 & 5.79 & 8.63 \\
					256 & 3.69 & 8.36 & 8.34 & 8.15 & 7.97 & 15.45 & 5.77 & {\bf 3.28} & 3.64 & 5.58 & 8.36 \\
					1024 & 3.75 & 8.07 & 8.06 & 7.88 & 7.71 & 15.45 & 5.76 & {\bf 3.30} & 3.67 & 5.50 & 8.07 \\
					\bottomrule
				\end{tabular}
	\vskip 0.1in
	\caption{Mean absolute error (in \%) of coverage for nonstationary linear regression.}
	\label{tab:ols-nonstationary}
\end{table}

\begin{figure}[h]
	\centering
	\begin{subfigure}{
			%\label{fig-combination-time-var1}
			\includegraphics[scale=0.45]{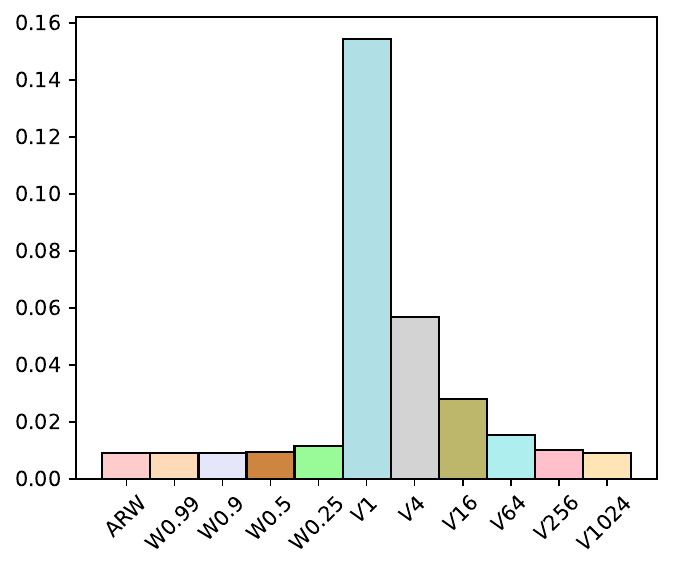}
		}
	\end{subfigure}
	\quad
	\begin{subfigure}{
			%\label{fig-combination-time-var10}
			\includegraphics[scale=0.45]{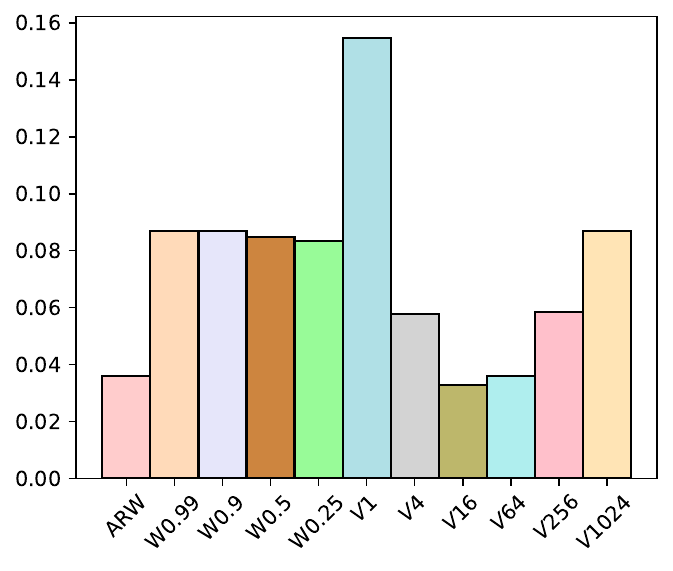}
		}
	\end{subfigure}
	\caption{Mean absolute errors of coverage for linear regression. Left: stationary case. Right: non-stationary case.}
	\label{fig:linreg-bar}
\end{figure}

\paragraph{Summary.} The results from mean estimation and linear regression show that $\mathrm{ARW}$ consistently achieves superior performances in both stationary and non-stationary settings.

\subsection{Real data}

Finally, we test our method using a real estate dataset maintained by the Dubai Land Department.\footnote{\url{https://www.dubaipulse.gov.ae/data/dld-transactions/dld_transactions-open}} 
We study the sales of apartments during the past 16 years (from January 1st, 2008 to December 31st, 2023).
There are 211,432 samples (sales) in total, and we want to predict the final price given characteristics of an apartment (e.g.,~number of rooms,~size,~location). Each week is treated as a time period, and there are 826 of them in total. 
\Cref{fig-housing} in \Cref{sec-housing-details} shows the weekly average of logarithmic prices, illustrating the distribution shift over time.
%On average, each week has around 250 samples. 
We split week $t$'s data into a training set $\datasetB_t^{\tr}$, a calibration set $\datasetB_t^{\ca}$, and a test set $\datasetB_t^{\rm te}$ with proportions $30\%$, $10\%$, and $60\%$, respectively. Most samples are allocated to the test set so that the coverage probability can be estimated accurately. We follow the standard practice to apply a logarithmic transform to the price and target that in our prediction. See \Cref{sec-housing-details} for other details of preprocessing.

For each $t \in \{ 10, 20, 30, \cdots, 820 \}$, we run XGBoost \citep{CGu16} on 4 different windows of historical training data $\{ \datasetB_j^{\tr} \}_{j=t - k\wedge t + 1  }^t$ with $k \in \{ 1, 16, 256, 1024 \}$, to get 4 point prediction models. For each of them, we compute its associated conformity scores on the calibration data $\{  \datasetB_j^{\ca} \}_{j=1}^t$, run \Cref{alg-predictive-inference} to obtain a prediction interval $\Chat_t$, and compute its coverage frequency on the test set $\datasetB_t^{\rm te}$. We also compare our adaptive rolling window approach to $\cV_k$ and $\cW_{\rho}$. We run the experiment independently over $100$ random seeds and report the average results in \Cref{table-housing-MAE,table-housing-AW}. For all of the $4$ training strategies, our approach consistently achieves the nominal coverage level with small average interval width. It outperforms weighting methods and is comparable to the best fixed-window method.

\begin{table}[h]
	\vskip 0.15in
	\centering
				\begin{tabular}{lccccccccccr}
					\toprule
					Training\\ Window & ARW & $\mathcal{W}_{0.99}$ & $\mathcal{W}_{0.9}$ & $\mathcal{W}_{0.5}$ & $\mathcal{W}_{0.25}$ & $\mathcal{V}_1$ & $\mathcal{V}_4$ & $\mathcal{V}_{16}$  & $\mathcal{V}_{64}$ & $\mathcal{V}_{256}$ & $\mathcal{V}_{1024}$ \\
					\midrule
					1 & 3.44  & 4.47 & 4.47 & 4.52 & 4.58 & 6.26 & 3.67 & {\bf 3.00} & 3.18 & 4.12 & 4.46  \\
					16 & 3.13 & 6.18 & 6.19 & 6.26 & 6.35 & 6.33 & 3.80 & {\bf 2.97} & 3.00 & 5.24 & 6.18  \\
					256 & 3.28 & 4.43 & 4.43 & 4.48 & 4.53 & 6.25 & 3.72 & {\bf 3.00} & 3.76 & 5.54 & 4.43 \\
					1024 & 3.10 & 3.63 & 3.63 & 3.64 & 3.68 & 6.34 & 3.73 & {\bf 3.06} & 3.29 & 3.84 & 3.63  \\
					\bottomrule
				\end{tabular}
	\vskip 0.1in
	\caption{Mean absolute error (in \%) of coverage on the housing data.}
		\label{table-housing-MAE}
\end{table}

\begin{table}[h]
	\vskip 0.15in
	\centering
				\begin{tabular}{lccccccccccr}
					\toprule
					Training\\ Window & ARW & $\mathcal{W}_{0.99}$ & $\mathcal{W}_{0.9}$ & $\mathcal{W}_{0.5}$ & $\mathcal{W}_{0.25}$ & $\mathcal{V}_1$ & $\mathcal{V}_4$ & $\mathcal{V}_{16}$  & $\mathcal{V}_{64}$ & $\mathcal{V}_{256}$ & $\mathcal{V}_{1024}$ \\
					\midrule
					1 & 1.58 & 1.70 & 1.70 & 1.70 & 1.71 & 1.37 & 1.47 & 1.49 & 1.55 & 1.65 & 1.70  \\
					16 & 1.10 & 1.41 & 1.41 & 1.42 & 1.43 & 0.98 & 0.98 & 0.97 & 1.09 & 1.30 & 1.41 \\
					256 & 0.98 & 1.15 & 1.15 & 1.16 & 1.17 & 0.97 & 0.99 & 0.97 & 0.96 & 0.98 & 1.15 \\
					1024 & 1.03 & 1.09 & 1.09 & 1.10 & 1.10 & 1.01 & 1.03 & 1.03 & 1.03 & 1.05 & 1.09 \\
					\bottomrule
				\end{tabular}
	\vskip 0.1in
	\caption{Mean interval width on the housing data.}
	\label{table-housing-AW}
\end{table}

%% file: main_discussions.tex
\section{Discussions}\label{sec-discussions}

We develop rolling window strategies for quantile estimation and predictive inference under temporal drift. Theoretical analyses and numerical experiments demonstrate their adaptivity to the unknown distribution shift. This work opens up a number of future directions. 
Our approach assumes decoupled training and validation procedures. It is worth investigating how our methods can be extended to procedures that make more efficient use of data, such as full conformal prediction \citep{VGS05}, jackknife$+$ and CV$+$ \citep{BCR21}. 
Another feature of our work is that no structural assumption on the temporal drift is needed. In practice, certain drift patterns can be learned from data or domain knowledge, such as seasonalities and trends. How to utilize those structures for better use of data while safeguarding against the risk of misspecification is of great importance.

%% file: appendix_proof.tex
\section{Proofs}\label{sec-proofs}

\subsection{Proof of \Cref{thm-bias-variance-decomp}}\label{sec-proof-thm-bias-variance-decomp}

Applying \Cref{lem-quantile} to the data $\{ \dataset_j \}_{j=t - k + 1}^t$ yields
\[
	\PP \bigg(
	| 
	F_t(
	\widehat{q}_{t, k} 
	)
	- (1 - \alpha) | \leq  \phi (t, k) +  \min_{\varepsilon > 0} \bigg\{
	\frac{
		\UB (B_{t, k}, \varepsilon, 1 - \alpha, \delta/2) 
	}{
		1 - \varepsilon
	} \bigg\}
	\bigg)
	\geq 1 - \delta ,
%\label{eqn-thm-Beta}
\]
where the function $\UB$ is defined in \Cref{defn-UB}. The proof is finished by taking $\varepsilon = 1/5$.

\subsection{Proof of \Cref{thm-bias-variance-decomp-proxy}}\label{sec-proof-thm-bias-variance-decomp-proxy}

Choose any $\delta \in (0, 1)$, $\varepsilon > 0$ and $\eta > 0$. Let $\UB(n, \varepsilon, \gamma, \delta) $ be defined as in \Cref{defn-UB} with $\gamma = 1 - \alpha$. We will show that with probability at least $1 - \delta$,
\begin{align}
	| 
	\widehat{F}_{t, i} (  \widehat{q}_{t, k} ) 
	- (1 - \alpha) | \leq 
	2 (1 + \eta)	\phi (t, k) + 
	\frac{
		1 + \eta
	}{
		1 - \varepsilon
	}
	\UB (B_{t, k}, \varepsilon, 1-\alpha, \delta/4) 
	+ \UB ( B_{t, i}, \eta, 1-\alpha, \delta / 4 ).
	\label{eqn-thm-double-hat}
\end{align}
If this is true, we can take $\eta = \varepsilon = 1/5$ and get
\begin{align*}
		| 
	\widehat{F}_{t, i} (  \widehat{q}_{t, k} ) 
	- (1 - \alpha) | 
	& \leq 
	%2 (1 + \eta)	\phi (t, k) + 
	%\frac{
		%	1 + \eta
		%}{
		%	1 - \varepsilon
		%}
	%\UB (B_{t, k}, \varepsilon, \delta/4) 
	%+ \UB ( B_{t, i}, \eta, \delta / 4 ) \\
	%& =
	\frac{12}{5} \phi (t, k) + \frac{3}{2}\UB (B_{t, k}, 1/5, 1-\alpha, \delta/4) + \UB ( B_{t, i}, 1/5, 1-\alpha, \delta / 4 ) \\
	& \le \frac{12}{5} \phi (t, k) +  \frac{6}{5} \psi (t, k, \delta/2) + \frac{4}{5} \psi (t, i, \delta/2).
\end{align*}

We now prove \eqref{eqn-thm-double-hat}. Let
\[
\psi_1  =
\frac{
	\UB (B_{t, k}, \varepsilon, 1-\alpha, \delta/4) 
}{
	1 - \varepsilon
}.
\]
Denote by $\cA$ the event that
\[
| 
F_t(
\widehat{q}_{t, k} 
)
- (1 - \alpha) | \leq  \phi (t, k) +  \psi_1  .
\]
According to \Cref{thm-bias-variance-decomp}, $\PP (\cA) \geq 1 - \delta / 2$. Let $r = \phi (t, k) +  \psi_1  $. When $\cA$ happens, we have
\begin{align*}
	\Qleft_{1 - \alpha - r} ( F_t ) \leq  \widehat{q}_{t, k} \leq \Qright_{1 - \alpha + r} ( F_t ) 
\end{align*}
and hence,
\[
	\PP \bigg(
	\widehat{F}_{t, i} \big( \Qleft_{ 1 - \alpha - r} ( F_t ) \big) 
	\leq \widehat{F}_{t, i} (  \widehat{q}_{t, k} ) \leq 
	\widehat{F}_{t, i} \big(
	\Qright_{ 1 - \alpha + r} ( F_t ) \big)
	\bigg)
	\geq 1 - \delta / 2.
	\label{eqn-thm-double-hat-1}
\]

Let $q^+ = \Qright_{1 - \alpha + r} ( F_t )$, which is deterministic. Then
\[\widehat{F}_{t, i} (q^+) = \frac{1}{B_{t, i}} \sum_{j = t - i + 1}^{t} \sum_{ \ell = 1 }^{ B_j } \bm{1} ( u_{j, \ell} \le q^+ ) 
	\quad\text{and}\quad 
	\EE \big[ \widehat{F}_{t, i} (q^+) \big]  = F_{t, i} (q^+)  .
\]
Note that $\widehat{F}_{t,i}(q^+)$ is an average of independent Bernoulli random variables. We will apply \Cref{lem-Bernstein} to obtain concentration of $\widehat{F}_{t,i}(q^+)$ around $F_{t,i}(q^+)$. By Assumption \ref{assumption-AC}, $F_t$ is absolutely continuous, so $F_t (q^+) = \min \{ 1 - \alpha + r , 1 \} = (1-\alpha) + \min \{ r , \alpha \}$. Hence,
\[
| \PP ( u_{j, \ell} \leq q^+  ) - (1-\alpha) | \leq | F_j(q^+) - F_t(q^+) | + \min \{ r , \alpha \} \leq  \phi (t, i) + r , \qquad t - i + 1 \leq j \leq t.
\]
Choose any $\delta \in (0, 1)$ and $\eta > 0$. By \Cref{lem-Bernstein},
\begin{equation}
\PP\left(
\widehat{F}_{t, i} (q^+) - F_{t, i} (q^+) 
\leq 
\eta [  \phi (t, i) + r ]  + 
\UB ( B_{t, i}, \eta, 1-\alpha, \delta / 4 ) \right)
\ge 1-\frac{\delta}{4}.
\label{eqn-thm-double-hat-2}
\end{equation}
Since
\[
F_{t, i} (q^+) \leq F_{t} (q^+)  + \phi (t, i) \leq  1 - \alpha + r +  \phi (t, i),
\]
then \eqref{eqn-thm-double-hat-2} implies
\[
\PP\left(\widehat{F}_{t, i} (\Qright_{q + r} ( F_t )) - (1 - \alpha)
	\leq 
	(1 + \eta) [ r + \phi (t, i) ] + \UB ( B_{t, i}, \eta, 1-\alpha, \delta / 4 ) \right)
	\ge 1-\frac{\delta}{4}.
\]

Similarly, we can prove the same tail bound for $- [ \widehat{F}_{t, i} (\Qleft_{q - r} ( F_t )) - (1 - \alpha) ]$. By union bounds, the following inequalities hold simultaneously with probability at least $ 1 - \delta / 2$:
\begin{align*}
	&  \widehat{F}_{t, i} (\Qright_{q + r} ( F_t )) - (1 - \alpha) 
	\leq 
	(1 + \eta) [ r + \phi (t, i) ] + \UB ( B_{t, i}, \eta, 1-\alpha, \delta / 4 ) , \\[4pt]
	&   \widehat{F}_{t, i} (\Qleft_{q - r} ( F_t )) - (1 - \alpha) \geq - \left\{
	(1 + \eta) [ r + \phi (t, i) ] + \UB ( B_{t, i}, \eta, 1-\alpha, \delta / 4 ) 
	\right\}.
\end{align*}
Combining this with \eqref{eqn-thm-double-hat-1}, we see that with probability at least $1 - \delta$,
\begin{align*}
	| 
	\widehat{F}_{t, i} (  \widehat{q}_{t, k} ) 
	- (1 - \alpha) | 
	& \leq 
	(1 + \eta) [	\phi (t, k) +  \psi_1    + \phi (t, i) ] 
	+ \UB ( B_{t, i}, \eta, 1-\alpha, \delta / 4 ) \\[4pt]
	& \leq (1 + \eta) \bigg(
	2 	\phi (t, k) + 
	\frac{
		\UB (B_{t, k}, \varepsilon, 1-\alpha, \delta/4) 
	}{
		1 - \varepsilon
	}
	\bigg)
	+ \UB ( B_{t, i}, \eta, 1-\alpha, \delta / 4 ).
\end{align*}
The last inequality is due to $i \leq k$ and thus $\phi (t, i) \leq \phi (t, k)$. We have obtained \eqref{eqn-thm-double-hat}.

\subsection{Proof of \Cref{thm-GL}}\label{sec-proof-thm-GL}
Let $\delta' = \delta/(4t^2)$. We will prove that with probability at least $1-\delta$,
\[
| F_t(\widehat{q}_{t,\widehat{k}}) - (1-\alpha) |  \le 6 \left[ \phi(t,k) + \psi(t,k,\delta') \right],\quad\forall k\in[t].
\]

By \Cref{thm-bias-variance-decomp} and the union bound,
\begin{equation}\label{eqn-proof-GL-event-bias-variance-decomp}
\PP\bigg( | F_t(\widehat{q}_{t, k} )	- (1 - \alpha) | \leq  \phi (t, k) + \psi ( t , k, \delta' ),\ \forall k\in[t] \bigg) \ge 1-t\delta'.
\end{equation}
By \Cref{thm-bias-variance-decomp-proxy} and union bounds,
\begin{equation}\label{eqn-proof-GL-event-phi-hat-dominance}
\PP \left( 0 \le \widehat{\phi}(t,k,\delta') \le \phi(t,k),\ \forall k\in[t] \right) \ge 1-t^2\delta'.
\end{equation}
For all $k\in[t]$, deterministically,
\begin{equation}\label{eqn-proof-GL-psi-dominance}
\psi(t,k,\delta'/2) \le 2\psi(t,k,\delta').
\end{equation}
From now on, assume that both the events in \eqref{eqn-proof-GL-event-bias-variance-decomp} and \eqref{eqn-proof-GL-event-phi-hat-dominance} happen. 

For all $k\in\{\widehat{k}+1,...,t\}$,
\begin{align}
|F_t(\widehat{q}_{t,\widehat{k}}) - (1-\alpha)|
&\le \phi(t,\widehat{k}) + \psi(t,\widehat{k},\delta')  \tag{by \eqref{eqn-proof-GL-event-bias-variance-decomp}}  \\[4pt]
&\le \phi(t,k) + \widehat{\phi}(t,\widehat{k},\delta') + \psi(t,\widehat{k},\delta') \tag{by \eqref{eqn-proof-GL-event-phi-hat-dominance}} \\[4pt]
&\le \phi(t,k) + \widehat{\phi}(t,k,\delta') + \psi(t,k,\delta') \tag{by the definition of $\widehat{k}$} \\[4pt]
&\le \phi(t,k) + \phi(t,k) + \psi(t,k,\delta') \tag{by \eqref{eqn-proof-GL-event-phi-hat-dominance}} \\[4pt]
&= 2 \phi(t,k) + \psi(t,k,\delta'). \label{eqn-proof-GL-large}
\end{align}
Next, take arbitrary $k\in[\widehat{k}]$. Then
\begin{align}
|\widehat{F}_{t,k}(\widehat{q}_{t,\widehat{k}}) - (1-\alpha)|
&\le 
\frac{12}{5} \widehat{\phi}(t,\widehat{k},\delta') + \frac{6}{5}\psi(t,\widehat{k},\delta'/2) + \frac{4}{5}\psi(t,k,\delta'/2) \tag{by the definition of $\widehat{\phi}$} \\[4pt]
&\le 
\frac{12}{5} \widehat{\phi}(t,\widehat{k},\delta') + \frac{12}{5}\psi(t,\widehat{k},\delta') + \frac{8}{5}\psi(t,k,\delta') \tag{by \eqref{eqn-proof-GL-psi-dominance}} \\[4pt]
&\le 
\frac{12}{5} \widehat{\phi}(t,k,\delta') + \frac{12}{5}\psi(t,k,\delta') + \frac{8}{5}\psi(t,k,\delta') \tag{by the definition of $\widehat{k}$} \\[4pt]
&\le 
\frac{12}{5} \phi(t,k) + 4\psi(t,k,\delta'). \tag{by \eqref{eqn-proof-GL-event-phi-hat-dominance}}
\end{align}
Let $r_k = \frac{12}{5} \phi(t,k) + 4\psi(t,k,\delta')$. Then $|\widehat{F}_{t,k}(\widehat{q}_{t,\widehat{k}}) - (1-\alpha)| \le r_k$ implies
\[
\Qleft_{1-\alpha - r_k} \big( \widehat{F}_{t,k} \big) \le \widehat{q}_{t,\widehat{k}} \le \Qright_{1-\alpha + r_k} \big( \widehat{F}_{t,k} \big).
\]
Therefore,
\begin{align*}
F_t(\widehat{q}_{t,\widehat{k}})
\le 
F_t\left( \Qright_{1-\alpha + r_k} \big( \widehat{F}_{t,k} \big) \right) 
&=
F_t\Bigg( \Qright_{1-\alpha + r_k} \bigg( \frac{1}{B_{t,k}} \sum_{j=t-k}^{t-1} \sum_{i=1}^{B_j} \delta_{u_{j,i}} \bigg) \Bigg) \\[4pt]
&=
\Qright_{1-\alpha + r_k} \bigg( \frac{1}{B_{t,k}} \sum_{j=t-k}^{t-1} \sum_{i=1}^{B_j} \delta_{ F_t( u_{j,i} ) } \bigg) \\[4pt]
&\le 
\Qright_{1-\alpha + r_k} \bigg( \frac{1}{B_{t,k}} \sum_{j=t-k}^{t-1} \sum_{i=1}^{B_j} \delta_{ F_j( u_{j,i} ) } \bigg) + \phi(t,k). \tag{by \Cref{lem-max-diff-order-stats} }
\end{align*}
By \Cref{lem-Beta-concentration} and the union bound,
\begin{multline}\label{eqn-proof-GL-event-quantile-concentration-1}
\PP \left( \Qright_{1-\alpha + r_k} \bigg( \frac{1}{B_{t,k}} \sum_{j=t-k}^{t-1} \sum_{i=1}^{B_j} \delta_{ F_j( u_{j,i} ) } \bigg)
\le 
( 1-\alpha + r_k ) + \frac{r_k}{4} + \frac{5}{4} \UB (B_{t,k},1/5,1-\alpha,\delta'),\ \forall k\in[t] \right) \\
\ge 1-t\delta',
\end{multline}
where $\UB(n, \varepsilon, \gamma, \delta) $ is defined in \Cref{defn-UB}. Similarly,
\begin{align*}
F_t(\widehat{q}_{t,\widehat{k}})
\ge 
F_t\left( \Qleft_{1-\alpha - r_k} \big( \widehat{F}_{t,k} \big) \right) 
&=
F_t\Bigg( \Qleft_{1-\alpha - r_k} \bigg( \frac{1}{B_{t,k}} \sum_{j=t-k}^{t-1} \sum_{i=1}^{B_j} \delta_{u_{j,i}} \bigg) \Bigg) \\[4pt]
&=
\Qleft_{1-\alpha - r_k} \bigg( \frac{1}{B_{t,k}} \sum_{j=t-k}^{t-1} \sum_{i=1}^{B_j} \delta_{ F_t( u_{j,i} ) } \bigg) \\[4pt]
&\ge 
\Qleft_{1-\alpha - r_k} \bigg( \frac{1}{B_{t,k}} \sum_{j=t-k}^{t-1} \sum_{i=1}^{B_j} \delta_{ F_j( u_{j,i} ) } \bigg) - \phi(t,k). \tag{by \Cref{lem-max-diff-order-stats} }
\end{align*}
By \Cref{lem-Beta-concentration} and the union bound,
\begin{multline}\label{eqn-proof-GL-event-quantile-concentration-2}
\PP \left( \Qleft_{1-\alpha - r_k} \bigg( \frac{1}{B_{t,k}} \sum_{j=t-k}^{t-1} \sum_{i=1}^{B_j} \delta_{ F_j( u_{j,i} ) } \bigg)
\ge 
( 1-\alpha-r_k ) - \frac{r_k}{4} - \frac{5}{4} R(B_{t,k},1/5,1-\alpha,\delta'),\ \forall k\in[t] \right) \\
\ge 1-t\delta'.
\end{multline}
When the events in \eqref{eqn-proof-GL-event-quantile-concentration-1} and \eqref{eqn-proof-GL-event-quantile-concentration-2} happen,
\begin{align}
| F_t(\widehat{q}_{t,\widehat{k}}) - (1-\alpha) | 
&\le 
\frac{5}{4} r_k + \frac{5}{4} R(B_{t,k},1/5, 1-\alpha, \delta') + \phi(t,k) \tag{by \eqref{eqn-proof-GL-event-quantile-concentration-1} and \eqref{eqn-proof-GL-event-quantile-concentration-2}} \\[4pt]
&\le
\left( 3 \phi(t,k) + 5\psi(t,k,\delta') \right) + \psi(t,k,\delta') + \phi(t,k) \notag \\[4pt]
&=
4\phi(t,k) + 6 \psi(t,k,\delta'). \label{eqn-proof-GL-small}
\end{align}

Therefore, when the events in \eqref{eqn-proof-GL-event-bias-variance-decomp}, \eqref{eqn-proof-GL-event-phi-hat-dominance}, \eqref{eqn-proof-GL-event-quantile-concentration-1} and \eqref{eqn-proof-GL-event-quantile-concentration-2} happen, which has probability at least $1-(3t+t^2)\delta' \ge 1-\delta$, we have \eqref{eqn-proof-GL-large} and \eqref{eqn-proof-GL-small}, and hence
\[
| F_t(\widehat{q}_{t,\widehat{k}}) - (1-\alpha) |  \le 6 \min_{k\in[t]} \left\{ \phi(t,k) + \psi(t,k,\delta')\right\}.
\]

%% file: appendix_technical.tex
\section{Technical lemmas}\label{sec-technical}

\begin{lemma}[Perturbation bound for empirical quantile]\label{lem-max-diff-order-stats}
Let $\{ x_i \}_{i=1}^n$ and $\{ y_i \}_{i=1}^n$ be real numbers. We have
\begin{align*}
&\max_{\gamma \in ( 0, 1) } \bigg|
\Qleft_{\gamma} 
\bigg(
\frac{1}{n}  \sum_{i=1}^{n} \delta_{ x_i }
\bigg)
- \Qleft_{\gamma} 
\bigg(
\frac{1}{n}  \sum_{i=1}^{n} \delta_{ y_i }
\bigg)
\bigg|
\leq \max_{i \in [n]} | x_i - y_i | .
\end{align*}
The same bound holds if we replace $\Qleft$ by $\Qright$. 
\end{lemma}

\begin{proof}[\bf Proof of \Cref{lem-max-diff-order-stats}]
Take permutations $(k_1,...,k_n)$ and $(\ell_1,...,\ell_n)$ of $[n]$ such that $x_{k_1}\le\cdots\le x_{k_n}$ and $y_{\ell_1}\le\cdots\le y_{\ell_n}$. Fix $\gamma\in(0,1)$. There exists $m\in[n]$ such that
\[
x_{k_m} = \Qleft_{\gamma} 
\bigg(
\frac{1}{n}  \sum_{i=1}^{n} \delta_{ x_i }
\bigg)
\quad\text{and}\quad
y_{\ell_m} = \Qleft_{\gamma} 
\bigg(
\frac{1}{n}  \sum_{i=1}^{n} \delta_{ y_i }
\bigg).
\]
Without loss of generality, assume $x_{k_m} \ge y_{\ell_m}$. By the pigeonhole principle, the set $S = \{k_i: m\le i\le n\} \cap \{\ell_i: 1\le i\le m\}$ is nonempty. Take $j\in S$, then $x_{k_m} \le x_j$ and $y_{\ell_m} \ge y_j$. Therefore,
\[
| x_{k_m} - y_{\ell_m} | = x_{k_m} - y_{\ell_m} \le x_j - y_j \le \max_{i\in[n]} |x_i - y_i|.
\]
Taking maximum over all $\gamma\in(0,1)$ finishes the proof.
\end{proof}

\begin{lemma}[Bernstein's bound with perturbation]\label{lem-Bernstein}
Let $\{ x_i \}_{i=1}^n$ be independent Bernoulli random variables, $p_i = \EE x_i$ and $\bar{p} = \frac{1}{n} \sum_{i=1}^{n} p_i$. For any $r > 0$, $\varepsilon > 0$, $\delta \in (0, 1)$ and $q \in [ 0, 1]$ that satisfies $|q - \bar{p}| \leq r$, the following holds with probability at least $1 - \delta$:
\[
\frac{1}{n} \sum_{i=1}^{n} x_i - \bar{p}
<  \varepsilon r + \sqrt{ \frac{ 2 q ( 1 - q ) \log (1 / \delta) }{ n } } + \bigg( \frac{2}{3} + \frac{1}{2 \varepsilon} \bigg) \frac{\log ( 1 / \delta ) }{ n }.
\]
The same upper bound also holds for $\bar{p} - \frac{1}{n} \sum_{i=1}^{n} x_i $ with probability at least $1 - \delta$.
\end{lemma}

\begin{proof}[\bf Proof of \Cref{lem-Bernstein}]
Below we only prove the upper bound on $\frac{1}{n} \sum_{i=1}^{n} x_i - \bar{p}$. The other bound follows from a similar argument. Let $\bar{x} = \frac{1}{n} \sum_{i=1}^{n} x_i$ and $\sigma^2 = \frac{1}{n} \sum_{i=1}^{n} \var (x_i) = \frac{1}{n} \sum_{i=1}^{n}  p_i (1 - p_i)$. By Lemma 3.1 in \cite{HHW24} (a version of Bernstein's inequality),
\[
\PP \left(
\frac{1}{n} \sum_{i=1}^{n}  x_i - \bar{p}
\leq
\sigma \sqrt{ \frac{ 2 \log (1 / \delta) }{ n } } + \frac{ 2 \log ( 1 / \delta ) }{ 3 n } 
\right) \geq 1 - \delta.
\]
Denote by $\cA$ the event on the left-hand side. By the concavity of the function $x \mapsto x(1-x)$ and Jensen's inequality,
\[
\sigma^2 = \frac{1}{n} \sum_{i=1}^{n}  p_i ( 1 - p_i) \leq   \bar{p} (1 - \bar{p}).
\]
When $\cA$ happens,
\[
\bar{x} - \bar{p} \leq
 \sqrt{ \frac{ 2  \bar{p} (1 - \bar{p}) \log (1 / \delta) }{ n } } + \frac{ 2 \log ( 1 / \delta ) }{ 3 n }.
\]
Let $h(x) = x(1-x)$. We have $|h'(x)| = |1 - 2 x| < 1$ for $x \in (0, 1)$. From $r > 0$ and $|q - \bar{p}| \leq r$, we get $\bar{p} (1-\bar{p}) < q(1-q) +r$ and
	\[
\sqrt{\bar{p} (1-\bar{p}) } < \sqrt{ q(1-q) + r } \leq \sqrt{ q (1-q) } + \sqrt{r}.
	\]
	We have
	\begin{align*}
		 \sqrt{  \frac{ 2 \bar{p} (1 - \bar{p}) \log (1 / \delta) }{ n } }
		& < 
		\sqrt{ \frac{  2 q (1-q) \log (1 / \delta) }{ n } }
		+  \sqrt{ \frac{ 2 r \log (1/ \delta) }{ n } } \\
		& \leq \sqrt{ \frac{  2 q (1-q) \log (1 / \delta) }{ n } }
		+ \frac{1}{2} \bigg(
		\varepsilon \cdot 2r + \frac{1}{ \varepsilon} \cdot \frac{ \log ( 1 / \delta) }{ n }
		\bigg)
 , \qquad \forall \varepsilon > 0.
	\end{align*}
Therefore, on the event $\cA$,
	\[
	\bar{x} - \bar{p} <  \varepsilon r + \sqrt{ \frac{ 2 q ( 1 - q ) \log (1 / \delta) }{ n } } + \bigg( \frac{2}{3} + \frac{1}{2 \varepsilon} \bigg) \frac{\log ( 1 / \delta ) }{ n }.
	\]
Consequently, the above inequality holds with probability at least $1  - \delta$.
\end{proof}

\begin{lemma}[Concentration of empirical quantile with perturbation]\label{lem-Beta-concentration}
Let $\{ x_i \}_{i=1}^n$ be i.i.d.~$\Unif (0, 1)$ random variables. For any $\gamma, \gamma', \delta \in (0, 1)$, with probability at least $1 - \delta$ we have
\[
\Qright_{\gamma} \bigg( \frac{1}{n}  \sum_{i=1}^{n} \delta_{ x_i } \bigg) - \gamma 
\leq
\min_{\varepsilon \in (0,1)} \bigg\{
(1 - \varepsilon)^{-1} \bigg[ \varepsilon | \gamma - \gamma' | + \sqrt{ \frac{ 2 \gamma' ( 1 - \gamma' ) \log (1 / \delta) }{ n } } + \bigg( \frac{2}{3} + \frac{1}{2 \varepsilon} \bigg) \frac{\log ( 1 / \delta ) }{ n } \bigg]
\bigg\} .
\]
The same upper bound also holds for $\gamma - \Qleft_{\gamma} \left( \frac{1}{n}  \sum_{i=1}^{n} \delta_{ x_i } \right)$ with probability at least $1-\delta$.
\end{lemma}

\begin{proof}[\bf Proof of \Cref{lem-Beta-concentration}]
Define $q^- = \Qleft_{\gamma} ( \frac{1}{n}  \sum_{i=1}^{n} \delta_{ x_i } )$ and $q^+ = \Qright_{\gamma} ( \frac{1}{n}  \sum_{i=1}^{n} \delta_{ x_i } )$. Let 
\[
T(\varepsilon) = (1 - \varepsilon)^{-1} \bigg[ \varepsilon |\gamma - \gamma'| +  \sqrt{ \frac{ 2 \gamma' ( 1 - \gamma' ) \log (2 / \delta) }{ n } } + \bigg( \frac{2}{3} + \frac{1}{2 \varepsilon} \bigg) \frac{\log ( 2 / \delta ) }{ n } \bigg],\quad\forall \varepsilon\in(0,1)
\]
be the function inside the $\min$ operator. For any $t > 0$, we have
\begin{align*}
&  q^+ - \gamma > t 
\quad\Rightarrow\quad
\frac{1}{n}  \sum_{i=1}^{n} \ind  (  x_i \leq  \gamma + t ) \le \gamma 
\quad\Leftrightarrow\quad
(\gamma + t) - \frac{1}{n}  \sum_{i=1}^{n} \ind  (  x_i \leq  \gamma + t )  \ge t ,
%\label{eqn-lem-quantile-concentration-1}
\\
& q^- - \gamma < - t 
\quad\Rightarrow\quad
\frac{1}{n}  \sum_{i=1}^{n} \ind  (  x_i \leq  \gamma - t ) \geq \gamma 
\quad\Leftrightarrow\quad
\frac{1}{n}  \sum_{i=1}^{n} \ind  (  x_i \leq  \gamma - t ) - (\gamma - t) \geq t .
%\label{eqn-lem-quantile-concentration-2}
\end{align*}
Based on the above,
%\begin{align*}
%\PP ( | \widehat{q} - \gamma | > t ) \leq \PP \bigg(
%(\gamma + t) - \frac{1}{n}  \sum_{i=1}^{n} \ind  (  x_i \leq  \gamma + t )  > t 
%\text{ or }
%\frac{1}{n}  \sum_{i=1}^{n} \ind  (  x_i \leq  \gamma - t ) - (\gamma - t) \geq t 
%\bigg)
%\end{align*}
%and
%\begin{multline*}
%	\PP ( | \widehat{q} - \gamma | \leq t ) \\ 
%	\geq
%	\PP \bigg(
%	(\gamma + t) - \frac{1}{n}  \sum_{i=1}^{n} \ind  (  x_i \leq  \gamma + t )  \leq t 
%	\text{ and }
%	\frac{1}{n}  \sum_{i=1}^{n} \ind  (  x_i \leq  \gamma - t ) - (\gamma - t) < t 
%	\bigg).
%\end{multline*}
\begin{align*}
& \PP (  q^+ - \gamma  \le t ) \ge \PP \bigg(
	(\gamma + t) - \frac{1}{n}  \sum_{i=1}^{n} \ind  (  x_i \leq  \gamma + t )  < t \bigg), \\
& \PP (  \gamma - q^-   \le t ) \ge \PP \bigg(
	\frac{1}{n}  \sum_{i=1}^{n} \ind  (  x_i \leq  \gamma - t ) - (\gamma - t) < t 
	\bigg).
\end{align*}

Choose any $\delta \in (0, 1)$ and $\varepsilon > 0$. Applying \Cref{lem-Bernstein} to $\{ \ind  (  x_i \leq  \gamma \pm t )   \}_{i=1}^n$, $p_i =  \gamma \pm t$ and $q = \gamma'$ yields that for every $\varepsilon>0$, each of the following inequalities holds with probability at least $1 - \delta$:
\begin{align*}
& (\gamma + t) - \frac{1}{n}  \sum_{i=1}^{n} \ind  (  x_i \leq  \gamma + t ) < 
\varepsilon t + \varepsilon |\gamma - \gamma'| + \sqrt{ \frac{ 2 \gamma' ( 1 - \gamma' ) \log (1 / \delta) }{ n } } + \bigg( \frac{2}{3} + \frac{1}{2 \varepsilon} \bigg) \frac{\log ( 1 / \delta ) }{ n }, \\[4pt]
& \frac{1}{n}  \sum_{i=1}^{n} \ind  (  x_i \leq  \gamma - t ) - (\gamma - t) < 
\varepsilon t + \varepsilon |\gamma - \gamma'| +  \sqrt{ \frac{ 2 \gamma' ( 1 - \gamma' ) \log (1 / \delta) }{ n } } + \bigg( \frac{2}{3} + \frac{1}{2 \varepsilon} \bigg) \frac{\log ( 1 / \delta ) }{ n }.
\end{align*}
Therefore, we would have $\PP (  q^+ - \gamma  \leq t ) \geq 1 - \delta $ and $\PP (  \gamma - q^- \leq t ) \geq 1 - \delta $ provided that
\begin{align*}
t \geq	\varepsilon t + \varepsilon|\gamma - \gamma'| + \sqrt{ \frac{ 2 \gamma ( 1 - \gamma ) \log (1 / \delta) }{ n } } + \bigg( \frac{2}{3} + \frac{1}{2 \varepsilon} \bigg) \frac{\log ( 1 / \delta ) }{ n } ,
\end{align*}
which is equivalent to $t \geq T(\varepsilon)$. Since the function $T$ is continuous on $(0,1)$ and $\lim_{\varepsilon \to 0^+} T(\varepsilon) = \lim_{\varepsilon \to 1^-} T(\varepsilon) = +\infty$, then $T$ attains its minimum value at some $\varepsilon_0\in(0,1)$. The proof is finished by taking $t = T(\varepsilon_0)$.
\end{proof}

Based on \Cref{lem-Beta-concentration}, we can study the concentration of empirical quantile under distribution shift. To state the result, we define a useful function that appears repeatedly in our proofs as probabilistic upper bounds.
\begin{definition}\label{defn-UB}
	For any $n \in \ZZ_+$, $\varepsilon > 0$ and $\gamma, \delta \in (0, 1)$, let
	\begin{align*}
		\UB (n, \varepsilon, \gamma, \delta) =   \sqrt{ \frac{ 2 \gamma ( 1 - \gamma ) \log (1 / \delta) }{ n } } + \bigg( \frac{2}{3} + \frac{1}{2 \varepsilon} \bigg) \frac{\log (1 / \delta ) }{ n } .
	\end{align*}
\end{definition}

\begin{lemma}[Concentration of empirical quantile under distribution shift]\label{lem-quantile}
Let $\{ x_i \}_{i=1}^n$ be independent random variables with absolutely continuous CDF's $\{ \cdf_i \}_{i=1}^n$, and denote by $\cdfhat$ their empirical CDF. Define $\phi = \max_{i\in[n]} \| \cdf_i - \cdf_n \|_{\infty}$. For any $\gamma, \delta \in (0, 1)$, we have
\begin{align*}
& \PP \left( \cdf_n \big( \Qright_{\gamma} ( \cdfhat ) \big) - \gamma \leq \phi + \min_{\varepsilon > 0} \bigg\{
\frac{
	\UB ( n, \varepsilon, \gamma, \delta) 
}{
	1 - \varepsilon
} \bigg\} \right) \geq 1 - \delta .
\end{align*}
The same bound also holds for $\gamma - \cdf_n \big( \Qleft_{\gamma} (\cdfhat) \big)$.
\end{lemma}

\begin{proof}[\bf Proof of \Cref{lem-quantile}]
We only prove the bound on $F_n \big( \Qright_{\gamma} ( \cdfhat ) \big) - \gamma $, as the other bound follows from the same analysis. By definition,
\begin{align*}
	& F_n \big( \Qright_{\gamma} ( \cdfhat ) \big) 
	= \Qright_{\gamma} \left(
	\frac{1}{n} \sum_{i=1}^{n}  \delta_{ F_n ( x_i ) }
	\right) .
\end{align*}
By absolute continuity, $\{ \cdf_i ( x_i ) \}_{ i=1 }^n$ are i.i.d.~$\Unif(0, 1)$ random variables. Hence, the distribution of $ \Qright_{\gamma}  (
\frac{1}{n} \sum_{i=1}^{n} \delta_{ F_i ( x_i ) }
)$ does not depend on $\{ F_i \}_{i=1}^n$. It can be viewed as a \emph{homogenized} version of $F_n \big( \Qright_{\gamma} ( \cdfhat ) \big) $. For any $\delta \in (0, 1)$, \Cref{lem-Beta-concentration} implies that with probability at least $1 - \delta$,
\begin{align}
\PP \left(
	\Qright_{\gamma} \bigg(
	\frac{1}{n} \sum_{i=1}^{n}  \delta_{ \cdf_i ( x_i ) }
	\bigg)
	- \gamma
	\leq  \min_{\varepsilon > 0} \bigg\{
	\frac{
		\UB ( n, \varepsilon, \gamma, \delta) 
	}{
		1 - \varepsilon
	} \bigg\} 
\right) \geq 1 - \delta.
	\label{eqn-homogenization}
\end{align}

Since $| F_n ( x_i ) - F_i ( x_i ) | \leq \| F_n - F_i \|_{\infty} \leq \phi$ a.s., then applying \Cref{lem-max-diff-order-stats} to $\{ F_n ( x_{i} ) \}_{i=1}^n$ and $\{ F_i ( x_{i} ) \}_{i=1}^n$ gives
\begin{align*}
	\bigg|
	F_n \big( \Qright_{\gamma} ( \cdfhat ) \big)
	-  \Qright_{\gamma} \bigg(
	\frac{1}{n} \sum_{i=1}^{n}  \delta_{ F_i ( x_i ) }
	\bigg)
	\bigg|
	\leq \phi  ,
	\qquad\text{a.s.}
\end{align*}
Combining this with \eqref{eqn-homogenization} completes the proof.
\end{proof}

%% file: appendix_experiments.tex
\section{Numerical experiments: additional details}\label{sec-experiments-details}

\subsection{A simplified algorithm for quantile estimation}\label{sec-algo-refined}

\begin{algorithm}[H]
	\begin{algorithmic}
		\STATE {\bf Input:} Datasets $\{ \dataset_j \}_{j=1}^t$ with $\dataset_j = \{ u_{j, i} \}_{i=1}^{B_j}$, miscoverage level $\alpha$ and hyperparameter $\delta'$.
		\STATE Let $m = \lceil\log_2 t\rceil + 1$. Let $k_s = 2^{s-1}$ for $s\in[m-1]$ and $k_m = t$.
		%\STATE Let $\delta' = \delta/(4t)$.
		\FOR{$s=1,...,m$}
		\STATE Compute $\widehat{q}_{t,k_s}$ according to \eqref{eqn-empirical-quantile}, and		
		\begin{align*}
			& \psi (t, k_s, \delta) = \sqrt{ \frac{\alpha ( 1 - \alpha) \log (1 / \delta)}{B_{t, k_s}} } + \frac{1}{B_{t, k_s}}, \\[4pt]
			&	\widehat{\phi}(t,k_s,\delta)
			= \frac{5}{12} \max\limits_{i\in[k]} \bigg( \big| \widehat{F}_{t,k_i}(\widehat{q}_{t,k_s}) - (1-\alpha) \big| - [   \psi (t,k_s,\delta) +  \psi (t,k_i,\delta) ]    \bigg)_+.
		\end{align*}
		
		\ENDFOR
		\STATE Choose any $	\widehat{s} \in \argmin_{ s\in[m] }  \{ \widehat\phi (t, k_s, \delta)  + \psi (t , k_s , \delta)  \}$.
		\RETURN $\widehat{q}_{t,k_{\widehat{s}}}$.
		\caption{Adaptive rolling window for quantile estimation (experiment version)}
		\label{alg-quantile-experiment}
	\end{algorithmic}
\end{algorithm}

\subsection{Details of the synthetic data experiment}\label{sec-random-function}

We give an outline of how the true mean sequence $\{\mu_t\}_{t=1}^T$ in the right panel of \Cref{fig:true-means} and the sequence $\{\beta_t\}_{t=1}^T$ for non-stationary linear regression are generated. 

We will construct a base sequence $\{u_t\}_{t=1}^T$, and set $\mu_t = 5u_t$ and $\beta_t = 2u_t$. The base sequence $\{u_t\}_{t=1}^T$ consists of $4$ parts, each representing a distribution drift pattern. In the first part, the sequence experiences large shifts. Then, it switches to a sinusoidal pattern. Following that, the environment stays stationary for some time. Finally, the sequence drifts randomly at every period, where the drift sizes are independently sampled from $\{1,-1\}$ with equal probability and scaled with a constant. 

The function for generating the base sequence takes in $3$ parameters $\texttt{N}$, $\texttt{n}$ and $\texttt{seed}$, where $\texttt{N}$ is the total number of periods, $\texttt{n}$ is the parameter determining the splitting points of the $4$ parts, and $\texttt{seed}$ is the random seed used for code reproducibility. In our experiment, we set $\texttt{N = 100}$,\ $\texttt{n = 2}$ and $\texttt{seed = 2024}$. The exact function can be found in our code at \url{https://github.com/eliselyhan/predictive-inference}.

\subsection{Details of the real data experiment}\label{sec-housing-details}

\Cref{fig-housing} shows the weekly average of logarithmic prices, illustrating the distribution shift over time.

\begin{figure}[!ht]
	\vskip 0.2in
	\begin{center}
		\includegraphics[scale=0.65]{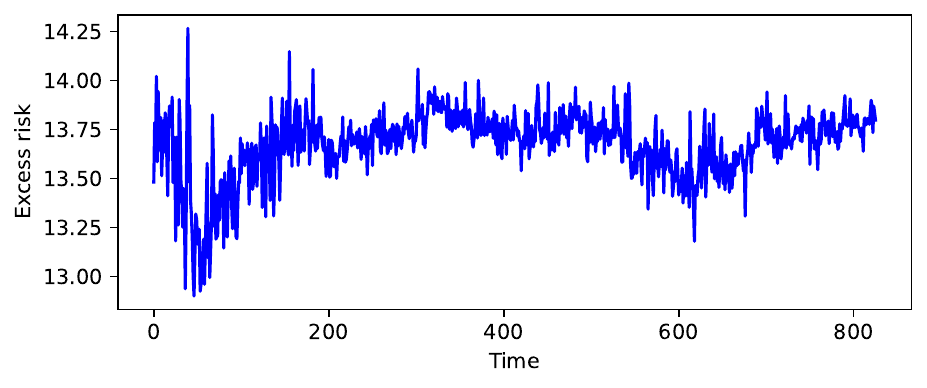}
		\caption{Weekly average of logarithmic prices.}
		\label{fig-housing}
	\end{center}
	\vskip -0.2in
\end{figure}

We focus on transactions of studios and apartments with 1 to 4 bedrooms, between January 1st, 2008 and December 31st, 2023. We import variables \texttt{instance\_date} (transaction date), \texttt{area\_name\_en} (English name of the area where the apartment is located in), \texttt{rooms\_en} (number of bedrooms), \texttt{has\_parking} (whether or not the apartment has a parking spot), \texttt{procedure\_area} (area in the apartment), \texttt{actual\_worth} (final price) from the data. 

We use \texttt{instance\_date} (transaction date) to construct weekly datasets. The target for prediction is the logarithmic of \texttt{actual\_worth}. The predictors are \texttt{area\_name\_en}, \texttt{rooms\_en}, \texttt{has\_parking} and \texttt{procedure\_area}. \texttt{area\_name\_en} has 58 possible values and encoded as an integer variable.

We remove a sample if its \texttt{actual\_worth} or \texttt{procedure\_area} is among the largest or smallest 2.5\% of the population, whichever is true. After the procedure, $91.6\%$ of the data remain.

We run XGBoost regression using the function \texttt{XGBRegressor} in the Python library \texttt{xgboost}. We set \texttt{random\_state} to be our random seed and do not change any other default parameters.